\documentclass[11pt,reqno]{amsart}
\usepackage{amsmath, amsthm, amssymb,stmaryrd}
\usepackage{enumitem}
\usepackage{extarrows}
\usepackage{bbm}

\usepackage{multirow}
\usepackage{fullpage}
\usepackage{adjustbox}
\usepackage[many]{tcolorbox}
\usepackage{xcolor}
\usepackage{wasysym}
\usepackage{url}
\usepackage{mathtools}
\usepackage{marvosym}
\usepackage[all,tips, 2cell]{xy}
\usepackage{amssymb}
\usepackage{mathrsfs}
\usepackage{array}
\usepackage{dsfont}
\usepackage{stackengine}

\usepackage{quiver}
\usepackage{tikz}
\usetikzlibrary{calc}
\usetikzlibrary{decorations,decorations.pathreplacing,decorations.markings,decorations.pathmorphing}
\tikzset{squiggly/.style={decorate, decoration=snake}}

\tikzset{super thick/.style={line width=3pt}}
\tikzstyle{far>}=[decoration={markings, mark=at position 0.75 with {\arrow{>}}}, postaction={decorate}]
\tikzstyle{mid>}=[decoration={markings, mark=at position 0.55 with {\arrow{>}}}, postaction={decorate}]
\tikzstyle{mid<}=[decoration={markings, mark=at position 0.55 with {\arrow{<}}}, postaction={decorate}]
\tikzset{super thick/.style={line width=3pt}}
\tikzstyle{far>}=[decoration={markings, mark=at position 0.75 with {\arrow{>}}}, postaction={decorate}]
\tikzstyle{mid>}=[decoration={markings, mark=at position 0.55 with {\arrow{>}}}, postaction={decorate}]
\tikzstyle{mid<}=[decoration={markings, mark=at position 0.55 with {\arrow{<}}}, postaction={decorate}]
\tikzstyle{knot}=[preaction={super thick, white, draw}]
\tikzstyle{coupon}=[draw, very thick, rectangle, rounded corners=5pt]
\tikzset{Rightarrow/.style={double equal sign distance,>={Implies},->},
triplecd/.style={-,preaction={draw,Rightarrow}},
quadruplecd/.style={preaction={draw,Rightarrow,
shorten >=0pt
},
shorten >=1pt,
-,double,double
distance=0.2pt}}
\tikzset{
    tripleline/.style args={[#1] in [#2] in [#3]}{
        #1,preaction={preaction={draw,#3},draw,#2}
    }
}
\tikzstyle{triple}=[tripleline={[line width=.15mm,black] in
      [line width=.7mm,white] in
      [line width=1mm,black]}] 
\tikzset{
    quadrupleline/.style args={[#1] in [#2] in [#3] in [#4]}{
        #1,preaction={preaction={preaction={draw,#4},draw,#3}, draw,#2}
    }
}
\tikzstyle{quadruple}=[quadrupleline={[line width=.3mm,white] in
      [line width=.6mm,black] in
      [line width=1.2mm,white] in
      [line width=1.5mm,black]}]

\definecolor{violet}{RGB}{148,0,211}
\definecolor{DarkGreen}{RGB}{0,150,0}

\usepackage[pdftex,plainpages=false,hypertexnames=false,pdfpagelabels]{hyperref}
\definecolor{medium-blue}{rgb}{0,0,.8}
\hypersetup{colorlinks, linkcolor={purple}, citecolor={medium-blue}, urlcolor={medium-blue}}
\newcommand{\arxiv}[1]{\href{http://arxiv.org/abs/#1}{\tt arXiv:\nolinkurl{#1}}}
\newcommand{\arXiv}[1]{\href{http://arxiv.org/abs/#1}{\tt arXiv:\nolinkurl{#1}}}

\newenvironment{proofsketch}{%
  \proof}{\endproof}

\usepackage{cleveref}

\DeclareMathOperator{\br}{br}

\DeclareMathOperator{\End}{End}

\DeclareMathOperator{\loc}{\mathrm{loc}}

\DeclareMathOperator{\Hom}{Hom}

\DeclareMathOperator{\pt}{pt}

\DeclareMathOperator{\Sq}{Sq}

\DeclareMathOperator{\Bimod}{\mathbf{Bimod}}

\newcommand{\Z}{\mathbb{Z}}

\DeclareMathOperator{\SH}{SH}

\newcommand{\Mod}{\mathbf{Mod}}

\newcommand{\Pic}{\mathcal{P}ic}

\newcommand{\Vect}{\mathbf{Vect}}
\newcommand{\sVect}{\mathbf{sVect}}
\newcommand{\Rep}{\mathbf{Rep}}

\newcommand{\sWitt}{s\mathscr{W}itt}
\newcommand{\Witt}{\mathcal{W}itt}

\newcommand{\tVect}{\mathbf{2Vect}}
\newcommand{\SVect}{\mathbf{SVect}}

\newcommand{\tSVect}{\mathbf{2SVect}}
\newcommand{\tRep}{\mathbf{2Rep}}
\newcommand{\sPic}{\mathscr{P}ic}
\newcommand{\sBrPic}{\mathscr{B}r\mathscr{P}ic}

\newcommand{\cZ}{\mathcal{Z}}
\newcommand{\sAut}{\mathscr{A}ut}
\newcommand{\rH}{\mathrm{H}}
\newcommand{\rB}{\mathrm{B}}

\newcommand{\bC}{\mathbb{C}}

\def\semicolon{;}
\def\applytolist#1{
    \expandafter\def\csname multi#1\endcsname##1{
        \def\multiack{##1}\ifx\multiack\semicolon
            \def\next{\relax}
        \else
            \csname #1\endcsname{##1}
            \def\next{\csname multi#1\endcsname}
        \fi
        \next}
    \csname multi#1\endcsname}

\def\calc#1{\expandafter\def\csname c#1\endcsname{{\mathcal #1}}}
\applytolist{calc}QWERTYUIOPLKJHGFDSAZXCVBNM;
\def\bbc#1{\expandafter\def\csname bb#1\endcsname{{\mathbb #1}}}
\applytolist{bbc}QWERTYUIOPLKJHGFDSAZXCVBNM;
\def\bfc#1{\expandafter\def\csname bf#1\endcsname{{\mathbf #1}}}
\applytolist{bfc}QWERTYUIOPLKJHGFDSAZXCVBNM;
\def\sfc#1{\expandafter\def\csname s#1\endcsname{{\sf #1}}}
\applytolist{sfc}QWERTYUIOPLKJHGFDSAZXCVBNM;
\def\fc#1{\expandafter\def\csname f#1\endcsname{{\mathfrak #1}}}
\applytolist{fc}QWERTYUIOPLKJHGFDSAZXCVBNM;
\def\rmc#1{\expandafter\def\csname rm#1\endcsname{{\mathrm #1}}}
\applytolist{rmc}QWERTYUIOPLKJHGFDSAZXCVBNM;
\def\scc#1{\expandafter\def\csname sc#1\endcsname{{\mathscr #1}}}
\applytolist{scc}QWERTYUIOPLKJHGFDSAZXCVBNM;
\def\scc#1{\expandafter\def\csname sc#1\endcsname{{\mathscr #1}}}
\applytolist{scc}QWERTYUIOPLKJHGFDSAZXCVBNM;

\numberwithin{equation}{section}

\theoremstyle{plain}
\newtheorem{thm}[equation]{Theorem}
\newtheorem*{thm*}{Theorem}
\newtheorem{cor}[equation]{Corollary}
\newtheorem{lem}[equation]{Lemma}
\newtheorem{prop}[equation]{Proposition}

\newtheorem*{claim*}{Claim}

\theoremstyle{definition}
\newtheorem{defn}[equation]{Definition}
\newtheorem{ansatz}[equation]{Ansatz}
\newtheorem*{trick*}{Trick}

\newtheorem*{convention*}{Convention}

\newtheorem{ex}[equation]{Example}

\newtheorem{rem}[equation]{Remark}

\newtheorem*{BosonicNondegen}{\Cref{BosonicNondegen}}
\newtheorem*{FermionicNondegen}{\Cref{FermionicNondegen}}

\newtheorem*{FermionicSET}{\Cref{prop:genuinelyfermionicTOclassification}}

\newtheorem*{SWitt}{\Cref{prop:sWittAnomaly}}

\newcommand\define[1]{\emph{#1}}

\title{The Classification of 3+1d Symmetry Enriched Topological Order}\author{Thibault D. Décoppet}
\address{Department of Mathematics, Harvard University,
1 Oxford St, Cambridge, MA 02138, USA}
\email{\href{mailto:decoppet@math.harvard.edu}{decoppet@math.harvard.edu}}

\author{Matthew Yu}
 \address{Mathematical Institute, University of Oxford, Woodstock Road, Oxford, OX2 6GG, UK}
 \email{\href{mailto:yumatthew70@gmail.com}{yumatthew70@gmail.com}}

\begin{document}

\begin{abstract}
We use a 2-categorical version of (de-)equivariantization to classify (3+1)d topological orders with a finite $G$-symmetry. In particular, we argue that (3+1)d fermionic topological orders with $G$-symmetry correspond to nondegenerate $\mathbf{2SVect}$-central $G$-crossed braided fusion 2-categories.
We then show that the categorical data necessary to define these theories agrees with that arising from a fermionic generalization of the Wang-Wen-Witten construction of bosonic topological theories with $G$-symmetry saturating an anomaly.
\end{abstract}

\maketitle

\section{Introduction}

A deeper mathematical understanding of higher categories has yielded novel physical insights into quantum field theory and  many adjacent areas, such as topological quantum field theory and quantum computing. In the case of fusion 2-categories, which are important for the symmetries of (2+1)d QFTs, a complete classification has been achieved in \cite{Decoppet:2024htz}. There, it was shown that beyond the theory of braided fusion 1-categories, all that is needed for the classification is data coming from group theory and cohomology. 

The study of fusion 2-categories has also led to progress in understanding topological order (TO), which are theories where all the operators are symmetry operators. 
In (3+1)d, both bosonic and fermionic topological orders have been classified in \cite{JF}, building on previous work \cite{Lan_2018,Lan_2019}. 
In (4+1)d, a classification was obtained in \cite{JFY}. Such results have applications to the analysis of categorical symmetries of quantum field theories. This is due to the development of the SymTFT, or topological holography picture \cite{Freed:2022qnc,kong2020algebraic,Wen:2013oza}, where the categorical symmetries
of a quantum field theory in $d$-dimensions is captured by a TQFT in $(d+1)$-dimensions. See \cite{Kaidi:2022cpf,Bhardwaj:2023wzd,Bhardwaj:2023ayw,Bhardwaj:2023bbf,Bartsch:2023wvv,Bhardwaj:2024xcx} for explicit examples where the higher categorical setup for the SymTFT is used.  Given this connection, it becomes greatly beneficial to further investigate higher-dimensional SymTFTs and their boundary conditions. Progress in this direction is essential for the study of symmetries in QFTs, particularly in the physically relevant case of (3+1)-dimensions.

Another important context in which TQFTs play a crucial role is as effective IR descriptions of UV gauge theories. In particular, when the UV theory possesses a 't Hooft anomaly, the corresponding TQFT can be used to match and saturate this anomaly in the IR. Studying anomalous TQFTs thus provides valuable insight into the nonperturbative dynamics and global structure of the UV theory.  Anomaly matching has recently led to significant progress in (2+1)d and (1+1)d QCD/QED, thanks to our understanding of candidate infrared TQFTs and their associated anomalies \cite{Gomis:2017ixy,Cordova:2017vab,Benini:2017dus,Komargodski:2017keh,Cordova:2017kue,Choi:2018tuh,Cordova:2018qvg,Delmastro:2021otj}. In order to extend these methods to (3+1)d fermionic gauge theories, a mathematical framework describing how the aforementioned TQFT can be coupled to a global symmetry, as well as the corresponding anomalies, is needed. In the presence of an anomaly, not all UV gauge theories necessarily flow to a gapped IR theory:\ there could exist symmetry-enforced-gaplessness constraints, we refer to \cite{CO1,CO2} for examples. Assuming that the UV theory does flow to a gapped, i.e.\ topological, theory in the IR, our formalism using fusion 2-categories describes the possible TQFTs.

Our main objective is to provide a unified framework that encapsulates (3+1)d topological orders and $G$-SETs, braided fusion 2-categories, and Lagrangian algebras for certain (4+1)d SymTFTs. This will be achieved by a 2-categorical version of (de-)equivariantization. We leverage the structure of braided fusion 2-categories to classify (3+1)d $G$-SETs in terms of group theory and cohomology.
Having done so, we also outline the classification of all braided fusion 2-categories. 

\subsection*{Notation}\label{section:notation}
We begin by compiling a few key terms related to
symmetries involving fusion 2-categories. Throughout we work over the complex numbers.

Associated to a fusion 2-category $\mathfrak{C}$, the braided fusion 1-category that describes the line operators is given by $\Omega\mathfrak{C} := \End_{\fC}(\mathds{1})$, where $\mathds{1}$ is the monoidal unit of $\mathfrak{C}$. 
The most fundamental invariant of the braided fusion 1-category $\Omega\mathfrak{C}$ is its so called M\"uger, or symmetric, center $\mathcal{Z}_{(2)}(\Omega \mathfrak{C})$. This is the subcategory of objects whose double braiding with any other object is trivial, i.e.\ the invisible line operators. For our classification purposes, we note that in the case when $\fC$ is braided, then $\Omega \fC$ is always symmetric. 

It follows from the work of Deligne \cite{deligne2002} that symmetric fusion 1-categories split into two classes:\ A symmetric fusion 1-category is Tannakian if it admits a fiber functor to $\mathbf{Vect}$. Otherwise, it is called super-Tannakian, and admits a fiber functor to $\mathbf{SVect}$. This dichotomy induces an analogous division for braided fusion 1-categories, and more generally for fusion 2-categories.

\begin{defn}\label{def:stronglyfusion}
    A fusion 2-category $\mathfrak{C}$ is \textit{bosonic} if $\mathcal{Z}_{(2)}(\Omega\mathfrak{C})$ is Tannakian. A fusion 2-category $\mathfrak{C}$ is said to be \textit{fermionic} otherwise.
\end{defn}

\begin{rem}
Physically speaking, a fusion 2-category is ``fermionic'', if there is an emergent fermion in the category. An emergent fermion is a fermionic line operator, however the presence of this operator does not mean that the theory requires a spin structure to be defined, and is hence still bosonic. The simplest example is when a theory has a fermionic symmetry, but can be realized on a bosonic Hilbert space. We will nevertheless use the term ``fermionic fusion 2-category'' to indicate a category with emergent fermions following the terminology introduced in \cite{JFY,D9}.
\end{rem}

Fusion 2-categories with no line operators also play a special role within the general theory of fusion 2-categories. We therefore recall the following terminology.

\begin{defn}\label{def:strongly}
A bosonic \textit{strongly fusion} 2-category is a fusion 2-category $\fC$ such that $\Omega \fC \cong \Vect$. A fermionic \textit{strongly fusion} 2-category is a fusion 2-category $\fC$ such that $\Omega \fC \cong \sVect$.
\end{defn}

\subsection{Results}\label{subseciton:mainresults}
We categorically define and classify  (3+1)d TOs with a finite $G$-symmetry, also known as (3+1)d $G$-symmetry enriched topological orders ($G$-SETs). We present the cohomological data needed to construct such TOs, and discuss their $G$-anomalies. Our description of anomalous (3+1)d $G$-SETs  is completely general and applies to any finite unitary symmetry $G$,\footnote{The definition of fusion 2-categories is able to accommodate mixing between $G$ and fermion parity, however the case of anti-unitary $G$ goes beyond what has currently been rigorously developed for fusion 2-categories. We expect that once the theory of unitary higher fusion categories is available, the procedure we use to enlarge a fermionic theory with unitary symmetry can also be generalized to include anti-unitary symmetries.} hence leading to a description of the (3+1)d gapped IR theories.\footnote{While the symmetry groups could differ between the UV to IR theories, the anomaly for symmetry in the IR must pull back to the anomaly in the UV. In our setup we assume that the IR has a finite $G$-symmetry.} 
This categorifies the work of \cite{Barkeshli:2014cna}, which we also generalize to include fermionic topological theories.
These theories, also known as (3+1)d fermionic SETs, have also been studied in \cite{Yang:2023gvi,Cheng:2024awi}, as the boundary of an SPT. See also \cite{Kong:2020jne} for a classification of (nonanomalous) SETs in lower dimensions.

Our classification proceeds by describing how a (3+1)d TO with no symmetries can be equipped with $G$-symmetry. Topological orders in (3+1)d with no symmetries fall into three cases:\ All the excitations are bosons, among the excitations there is an emergent fermion, and among the excitations  there is a local fermion. Theories of each type have been classified by \cite{Kong:2014qka,kong2015boundary,kong2017boundary,Lan_2018,Lan_2019,JF}. 
We present an alternative construction of both all-boson topological orders and emergent-fermion topological orders via a 2-categorical analogue of (de-)equivariantization. In the 1-categorical setting, (de-)equivariantization exchanges a 0-form symmetry, given by the action of a finite group $G$ on a topological order, with a dual topological order carrying an action of the dual symmetry $\Rep(G)$. This correspondence can be viewed as arising from a particular gauging procedure, which may produce non-invertible dual symmetries. More specifically, it gives an equivalence of categories between the linear categories with $G$-action, and linear categories with $\Rep(G)$ action \cite{DGNO}. 
An analogous relationship holds in the 2-categorical setting.
We emphasize a classification strategy of (3+1)d TOs that involves (de-)equivariantization, as it will also play a crucial role in the classification of (3+1)d $G$-SETs. 
This approach therefore provides a holistic framework for understanding all (3+1)d TOs, i.e.\ nondegenerate braided fusion 2-categories, both with and without global symmetry.

\subsubsection{Nondegenerate Braided Fusion 2-Categories}\label{subsub:introNondegen}

On one hand, if $\fB$ is a bosonic braided fusion 2-category, then $\Omega \fB$ will be a Tannakian symmetric fusion 1-category, say $\mathbf{Rep}(G)$, for some finite group $G$. On the other hand, if $\fB$ is a fermionic braided fusion 2-category, then $\Omega \fB$ is a super-Tannakian symmetric fusion 1-category of the from $\Rep(\widetilde{G},z)$, for some finite super-group $(\widetilde{G},z)$. In particular, it contains $\Rep({G})$ with $G=\widetilde{G}/z$ as a symmetric tensor fusion sub-1-category. In either case, we find that there is a canonical braided tensor 2-functor $\mathbf{2Rep}({G})\rightarrow\mathfrak{B}$, which we can use to de-equivariantizing $\mathfrak{B}$. This produces a (not necessarily faithfully graded) $G$-crossed braided strongly fusion 2-categories, which is either bosonic or fermionic depending on the input $\mathfrak{B}$. Moreover, this procedure can be reversed -- this is equivariantization -- so that we can recover $\mathfrak{B}$ from its corresponding $G$-crossed braided strongly fusion 2-category.

We present two theorems classifying \textit{faithfully} graded $G$-crossed braided strongly fusion 2-categories. By equivariantization, this yields classifications of both bosonic and fermionic nondegenerate braided fusion 2-categories, which we will presently state. In the bosonic case, the result is well-known.

\begin{BosonicNondegen}[\cite{Lan_2018,JF}]\label{thm:nondegenBosonic}
    Nondegenerate bosonic braided fusion 2-categories are classified by a finite group $G$ and a class $\pi \in \rH^4(\rB G; \mathbb{C}^\times)$.
\end{BosonicNondegen}

\begin{rem}\label{rem:recoverBosonic}
    In particular, nondegenerate bosonic braided fusion 2-categories are of the form $\mathcal{Z}(\tVect_G^{\pi})$, that is, they arise as the Drinfeld center of the bosonic strongly fusion 2-category $\tVect_G^{\pi}$. The braided fusion 2-categories $\mathcal{Z}(\tVect_G^{\pi})$ were extensively studied in \cite{KTZ}. Physically, these categories classify all (3+1)d topological orders where the excitations are all bosons.
\end{rem}

We now turn to (3+1)d topological orders where the spectrum contains an emergent fermion, though the theory is itself still bosonic. These are described by fermionic nondegenerate braided fusion 2-categories, whose classification requires a generalized cohomology theory denoted $\SH^{*+\varpi}(\rB G)$. This cohomology theory goes by the name of $\varpi$-twisted (extended) supercohomology, and is associated to the space $\tSVect^\times$ of invertible objects and morphisms in $\tSVect$. We refer the reader to \S\ref{subsection:perlimFTO} for additional details.

\begin{FermionicNondegen}\label{thm:nondegenFermionic}
    Nondegenerate fermionic braided fusion 2-categories are classified by a finite group $G$, a class $\varsigma \in \SH^{5+\kappa}(\rB^2 \Z/2)$, where $\kappa$ is the nontrivial class in $\rH^2(\rB^2 \Z/2;\Z/2)$, a class $\tau\in \rH^2(\rB G;\Z/2)$, such that $\varsigma \circ \tau$ is trivial in $\SH^{5+\tau}(\rB G)$, and a class $\varpi \in \SH^{4+\tau}(\rB G)$.
\end{FermionicNondegen}
\begin{rem}\label{rem:recoverFermionic}
The last result provides a more explicit description of the data given in \cite[Corollary V.5]{JF} and refines a result of \cite{Lan_2019}.
Namely, in the case where $G$ is trivial, we recover the theories $\mathcal{S}$ and $\mathcal{T}$ investigated in \cite{JF2,JFR}, which are two nondegenerate fermionic braided fusion 2-categories $\fB$ with $\Omega\fB = \SVect$ and $\pi_0 \fB =\Z/2$. They are distinguished by the class $\varsigma\in\SH^{5+\kappa}(\rB^2 \Z/2)\cong\mathbb{Z}/2$. If $\varsigma$ is trivial, the corresponding theory, $\mathcal{S}$, is described by the nondegenerate braided fusion 2-category $\mathcal{Z}(\tSVect)$. If $\varsigma$ is nontrivial, the corresponding theory, $\mathcal{T}$, was overlooked in \cite{Lan_2019}.

More generally, if $\varsigma$ is trivial, but $G$ is arbitrary, our classification recovers all the nondegenerate braided fermionic fusion 2-categories of the form $\cZ(\tSVect^\varpi_G)$, that is, those arising as the Drinfeld center of a fermionic strongly fusion 2-category. However, it is known that distinct fermionic strongly fusion 2-categories can have equivalent Drinfeld center \cite{D9}.
This subtle point is addressed in \Cref{rem:redundancies} below, to which we direct the interested reader.
For a different, more physically minded, approach to this question, we refer the reader to \cite{Teixeira:2025qsg}.
\end{rem}

\begin{rem}
The above classification of nondegenerate braided fusion 2-categories also leads to a classification of Lagrangian algebras in certain fusion 3-categories of interest. More specifically, in \S\ref{subsection:lagrangianalg}, we present a classification of Lagrangian algebras in $\cZ(\mathbf{3Vect}_G)$, which has found application in \cite{Antinucci:2025fjp}. Our results thus lay the mathematical foundations necessary to push the classification of (3+1)d phases forward. Namely, the classification of Lagrangian algebras in $\cZ(\mathbf{Vect}_G)$ and $\cZ(\mathbf{2Vect}_G)$ has been used in the categorical Landau paradigm in (1+1)d and (2+1)d \cite{Bhardwaj:2023idu,Bhardwaj:2023fca,Bhardwaj:2024qrf,Bhardwaj:2024qiv,Bhardwaj:2025piv,Wen:2023otf,Wen:2024qsg,Wen:2025thg}.
\end{rem}

We recover the classification of genuinely fermionic (3+1)d topological orders given in \cite[Corollary V.4]{JF}. Mathematically, this corresponds to classifying nondegenerate $\mathbf{2SVect}$-central braided fusion 2-categories. Below, we explain how to classify these genuinely fermionic topological orders in the presence of a $G$-symmetry.

\subsubsection{(3+1)d $G$-SETs and their Anomalies} \label{subsub:classifyingSET}

We now focus on (3+1)d $G$-SETs in the strictly fermionic setting. By this, we mean that the corresponding TO contains local fermions and requires a spin structure to be defined. The classification in the bosonic case follows similarly, but is simpler and only requires removing a few adjectives from the fermionic classification. We refer the reader to \S\ref{section:G-TQFTs} for details.
In the work \cite{JF}, it was shown that fermionic (3+1)d topological orders are classified by nondegenerate $\mathbf{2SVect}$-central braided fusion 2-categories. That is, a braided fusion 2-category $\fB$ equipped with a sylleptic monoidal 2-functor $\tSVect \to \cZ_{(2)}(\fB)$. Here, the notation $\cZ_{(2)}(\fB)$ refers to the \textit{sylleptic center} of the braided fusion 2-category $\fB$, as defined in \cite{Cr}. This is a 2-categorical version of the M\"uger center defined for braided fusion 1-categories.

Equipping a fermionic TQFT with a finite $G$-symmetry amounts, at the level of 2-categories, to building a $G$-crossed braided extension. For a more details on how extensions encode the properties of symmetry actions, fractionalization patterns, and symmetry defects in (2+1)d, see \cite{Barkeshli:2014cna,Barkeshli:2019vtb,Bulmash:2020flp,Manjunath:2020kne,Barkeshli:2021ypb}.
Fermionic (3+1)d topological orders with $G$-symmetry correspond more precisely to $\mathbf{2SVect}$-central $G$-crossed braided fusion 2-categories satisfying a nondegeneracy condition -- namely, we require that the sylleptic 2-functor $\mathbf{2SVect}\rightarrow\mathcal{Z}_{(2)}(\mathfrak{B})$ is an equivalence. Using a version of (de-)equivariantization, these $G$-crossed braided extensions can be explicit classified. 

\begin{FermionicSET}\label{thm:GTQFT}
 Fermionic (3+1)d topological orders with $G$-symmetry are classified by a finite group $H$ surjecting onto $G$ together with a class in $\SH^{4}(\rB H)$.\footnote{We only consider topological orders with a trivial vacuum.}
\end{FermionicSET}

Given a fermionic (3+1)d topological order with a $G$-action, we can ask whether it is possible to insert $G$-defects into the $\mathbf{2SVect}$-central braided fusion 2-category $\fB$ such that the fusion and associativity relations
respect the group multiplication of $G$.
We will call the obstruction to being able to perform such a $\mathbf{2SVect}$-central $G$-crossed braided extension of $
\fB$ as the anomaly to gauging the $G$-symmetry, or, colloquially, the ``$G$-anomaly''. There is a generalized cohomology theory, which we denote by $\mathrm{SW}^*$, that classifies $G$-anomalies, and is represented by the 4-groupoid $\rB \sWitt$.

The set of isomorphism classes of objects in this 4-groupoid is the super-Witt group, introduced in \cite{DNO}, of Witt equivalence classes of nondegenerate $\mathbf{SVect}$-central braided fusion 1-categories, i.e.\ (genuinely) fermionic (2+1)d topological orders up to gapped boundary.
This 4-groupoid has been considered in \cite{JF2}, where it is argued that its homotopy groups classify fermionic topological orders.
Now, the fact that the space $\rB \sWitt$ describes the anomaly to gauging a $G$-symmetry on the $\mathbf{2SVect}$-central braided fusion 2-category $\fB$ is captured by the following fiber sequence:
\begin{equation}\label{eq:superfibseq}
    \begin{tikzcd}
\mathrm{B}\mathscr{S}\!\mathscr{P}ic(\mathfrak{B}) \arrow[r] & \mathrm{B}\mathscr{A}ut^{\br}_{\mathbf{2SVect}}(\fB) \arrow[r, "{[-]}"] & \rB \sWitt\,.
\end{tikzcd}
\end{equation}
Namely, a $G$-action on $\fB$ is precisely the data of a map of spaces $\rB G\rightarrow \rB \mathscr{A}ut^{\br}_{\mathbf{2SVect}}(\fB)$ into the delooping of the space of $\mathbf{2SVect}$-central braided autoequivalences of $\fB$. Moreover, $G$-defects can be inserted into $\mathfrak{B}$ if and only if this map of spaces can be lifted to $\mathrm{B}\mathscr{S}\!\mathscr{P}ic(\mathfrak{B})$, the delooping of the $\mathbf{2SVect}$-central Picard space of $\fB$. 
This yields the following mathematical result.

\begin{SWitt}\label{thm:obstruction}
 A $\mathbf{2SVect}$-central braided fusion 2-category $\mathfrak{B}$ equipped with a $G$-action $\rho:\rB G\rightarrow \mathrm{B}\mathscr{A}ut^{\br}_{\mathbf{2SVect}}(\fB)$ can be extended to a nondegenerate $\mathbf{2SVect}$-central $G$-crossed braided fusion 2-category if and only if $[\rho]\in \mathrm{SW}^5(\rB G)$, the anomaly of the action, is trivial.
\end{SWitt}

\noindent Physically speaking, obstructions to constructing fermionic (3+1)d topological orders with $G$-symmetry (aka fermionic (3+1)d $G$-SETs) are classified by homotopy classes of maps from $\rB G$ to the space $\rB \sWitt$. Therefore, the anomalies for fermionic (3+1)d topological orders with $G$-symmetry are captured by $\mathrm{SW}^5(\rB G)$.\footnote{The group that classifies anomaly of fermionic (3+1)d topological orders with $G$-symmetry can also be derived using the general definition of anomalies for TQFTs in \cite{Stockall:2025ppu}.}

In the works of \cite{Witten:2016cio,Wang:2017loc}, the authors gave a path integral construction of a $d$-dimensional bosonic TQFT that saturates an anomaly $\omega \in \rH^{d+1}(\rB G;\mathbb C^\times)$. The data presented there should generalize to allow for the construction of fermionic TQFTs, in which case a natural cocycle construction would involve a supercohomology class $\varpi \in \SH^{d+1}(\rB G)$. However, a path integral (i.e.\ state sum) construction for a boundary fermionic TQFT that saturates a fully general anomaly $\varpi$ has not yet been achieved. Some progress has been made in \cite{Kobayashi:2019lep}, where the authors give a state sum for the TQFT saturating $\varpi$ when two of its layers are nontrivial, i.e.\ when $\varpi$ is actually a class in \emph{restricted} supercohomology \cite{Freed:2006mx,Gu:2012ib}.
    
Even though the most general state sum construction is currently out of reach, we can nevertheless describe the TQFT abstractly from the point of view of fusion 2-categories. We conclude in \S\ref{subsection:GTQFT} and \S\ref{subsection:obstructiongauging}  that the data of Theorems \ref{prop:application} and \ref{prop:sWittAnomaly} matches the expected generalization of the data from \cite{Wang:2017loc}. Thence, we conclude that the two methods of constructing (3+1)d TQFTs with anomalies give the same result. The following ansatz can be taken as parallel to the data used by Wang-Wen-Witten \cite{Wang:2017loc} for their construction:
\begin{ansatz}[Fermionic Wang-Wen-Witten]\label{ansatz:fernionicTQFT}
 A (3+1)d fermionic topological theory with $G$-symmetry and {$\mathrm{SW}^5(\rB G)$} anomaly is realized as a $K$-gauge theory from the following data:
\begin{itemize}
    \item A short exact sequence $1 \rightarrow K \rightarrow H \rightarrow G\rightarrow 1$, where the normal subgroup $K$ is not necessarily abelian;
     \item A class in $\SH^4(\rB H)$;
    \item A class in $\SH^4(\rB K)$ giving the Dijkgraaf-Witten action for the $K$-gauge theory. 
\end{itemize}
\end{ansatz}
\begin{rem}
For the $K$-gauge theory to actually realize a specific cocycle representing a fixed anomaly, there are a couple more steps. Following \cite{Wang:2017loc}, the first is to compute the group $\mathrm{SW}^5(\rB G)$ and find its generators. The next step is to find a group $H$ which fits into an extension of the form 
\begin{equation}
    1\longrightarrow K \longrightarrow H \longrightarrow G \longrightarrow 1\,,
\end{equation}
where the anomaly vanishes when pulled back to $\mathrm{SW}^5(\rB H)$. Then one needs to solve a differential equation for the degree 4-cochain that witnesses the trivialization. Since $\rB\sWitt$ has a layer that goes beyond supercohomology, if one is to  follow \cite{Wang:2017loc} exactly, then one might want to restrict to a setting where the anomaly is just captured by a supercohomology cocycle. In particular, since the gauge theory that realizes the anomaly will be defined on a degree 4 supercohomology class, it is hard to see how it can accommodate something coming from $\pi_1(\rB \sWitt)$, the super-Witt group, as part of its anomaly. We give a short discussion on the ``beyond cocycle'' part of the obstruction at the end of \S\ref{subsection:obstructiongauging}. The cocycle approach to anomalies for constructing (3+1)d TQFTs with anomalies will be explored in more details for some physically relevant groups in future work \cite{DYY1}.
\end{rem}

\section{Preliminaries}\label{section:prelim}

\subsection{Bosonic Topological Orders}\label{subsection:BTQFT}
By a topological order we mean a topological field theory that is fully determined by its algebra of extended topological operators and satisfies remote detectability, i.e., such that there are no completely invisible operators, that is, operators linking trivially with all other operators. This leads to the following definition.
\begin{defn}\cite[Definition I.1]{JF}\label{def:toporder}
     A topological order in ($n$+1)-dimensions is a multifusion $n$-category with trivial center.\footnote{We will use topological order synonymously for anomalous fully extended TQFTs in the sense of \cite{Freed:2012bs}.}
\end{defn}
The trivial center can be thought of as a nondegeneracy condition on the category that describes the topological order. 

\begin{rem}
    This notion of topological order that includes remote detectability is described by a \textit{pure state} that is the ground state of a local, gapped Hamiltonian. Such a theory has no influence from the environment, and therefore is not subject to local decoherence, which could turn the pure state into a mixed state, hence destroying the topological order.
\end{rem}

Although the definition of topological order allows for multiple ground states, our discussion will focus solely on the \textit{fusion} case, thereby restricting local operators to the vacuum operator.\footnote{By fusion we mean that the vacuum is a simple rather than decomposable object. Furthermore, there are no nontrivial local operators, which can implement morphisms between line operators, thanks to the semisimplicity condition inherent to the notion of a higher (multi)fusion category.} This condition along with the principle of remote detectability implies that all codimension-1 operators arise as condensation descendant. This reduces the problem of classifying topological orders in any dimension down by one categorical level.

We now review the ``all bosons'' classification of topological orders in (3+1)d following \cite{Lan_2018}. The classification involving emergent fermions as well as local fermions unfolds in a similar way except for the fact that that super-vector spaces are involved \cite{Lan_2019,JF}. See also \cite{Thorngren:2020aph} for an explicit construction of the TQFT in the case when it contains an emergent fermion. The first step in the classification is to condense all the line operators of the theory.
For a (3+1)d topological order described by a fusion $3$-category $\mathbf{A}$, we write $\Omega \mathbf{A}: = \End_{\mathbf{A}}(\mathds{1})$ where $\mathds{1}$ is the monoidal unit in $\mathbf{A}$, for the braided fusion 2-category of surface defects. We have an adjoint map to $\Omega$ that we denote $\Mod(-)$, i.e.\ taking the category of modules, or condensation defects. As a shorthand, we will use $\Mod^n(-)$ to denote taking modules $n$-times. By this we mean, iteratively delooping and Karoubi completing $n$-times, following the prescription of \cite{Gaiotto:2019xmp,JF}. For us, we will consider at most $n=2$ and we will not address the technicalities arising when $n$ is larger. We refer the interested reader to \cite{Bhardwaj:2024xcx} for some comments in this direction.

The fusion 1-category of line defects in $\mathbf{A}$ is given by $\Omega^2\mathbf{A}$, which is a symmetric fusion 1-category. Physically one should think of lines as being able to move in an ambient (3+1)d setting without braiding. Since we are working in the bosonic setting, we can choose a fiber functor $F: \Omega^2\mathbf{A} \rightarrow \Vect$, and suspend it to $\Mod^2(F): \Mod^2(\Omega^2 \mathbf{A})\rightarrow \Mod^2 (\Vect) = \mathbf{3Vect}$ where $\Mod^2(\Omega^2 \mathbf{A})$ is the sub-3-category of $\mathbf{A}$ of operators arising as condensation descendants of line operators. Moreover, the functor $\Mod^2(F)$ turns $\mathbf{3Vect}$ into a module for $\Mod ^2(\Omega^2\mathbf{A})$. Taking the base change of this module along the inclusion $\Mod ^2(\Omega^2\mathbf{A}) \subset \mathbf{A}$ produces a module $\mathbf{M}$ for $\mathbf{A}$ given in formulas by:
\begin{equation}
    \mathbf{M} := \mathbf{A} \boxtimes_{\Mod^2 (\Omega^2\mathbf{A})} \mathbf{3Vect}.
\end{equation}
We then consider the fusion 3-category $\mathrm{End}_{\mathbf{A}}(\mathbf{M})$ of $\mathbf{A}$-linear endomorphism of $\mathbf{M}$. By definition, $\mathbf{M}$ witnesses a Morita equivalent between $\mathbf{A}$ and $\mathrm{End}_{\mathbf{A}}(\mathbf{M})$. But the category $\mathrm{End}_{\mathbf{A}}(\mathbf{M})$ has no nontrivial line operators by construction. By remote detectability, this implies that all codimension two operators arise as condensation descendants \cite{JF}, see also \cite{Lan_2018,Lan_2019}. 

In summary, all of the operators in $\mathrm{End}_{\mathbf{A}}(\mathbf{M})$ arise as condensation from the vacuum. This implies that any bosonic (3+1)d topological order arises from gauging a $G$-SPT set by the fiber functor $F$, since the fiber functor produces a gapped domain wall with $G$-symmetry. Hence, such theories are Dijkgraaf-Witten theories with Lagrangian description given by a class in $\rH^4(\rB G ;\mathbb C^\times)$, as first observed in \cite{Lan_2018}. This spells out the Physics associated to the classification in \Cref{rem:recoverBosonic}. We proceed to give a more detailed account of the categorical structure of bosonic (3+1)d TOs, which will be necessary for our main results.

\subsection{Fusion 2-Categories and their Braidings}

We use the definition of fusion 2-categories provided in \cite{douglas2018fusion}. More precisely, a fusion 2-category is a finite semisimple 2-category equipped with a rigid monoidal structure, whose monoidal unit is simple.

Given a monoidal 2-category $\fB$ with tensor product $\otimes$, a \textit{braiding} on $\fB$ consists of a natural equivalence 
\begin{equation}\label{eq:braidingF2C}
    b_{A|B}:A \otimes B \xrightarrow{\simeq} B \otimes A,
\end{equation}
for every object $A$, $B$ of $\mathfrak{B}$,
together with \define{hexagonators} $R_{(A|-,-)}$ and $S_{(-,-|A)}$ satisfying various axioms recorded for instance in \cite{SP} and \cite[Section 2.1.1]{Decoppet:2022dnz}.

We will be specifically interested in the case when $\fB$ is a braided strongly fusion 2-category, as recalled in Definition \ref{def:stronglyfusion}. Under this additional hypothesis, the connected components of $\mathfrak{B}$ form an abelian group $E$ as all the simple objects of $\mathfrak{B}$ are invertible by \cite{JFY}. It then follows easily, see for instance \cite[Section 2.2]{JFY2}, that braided bosonic strongly fusion 2-categories are classified by the finite abelian group $E$ together with a class in $\rH^5(\rB^2 E; \mathbb{C}^\times)$. The classification of braided fermionic strongly fusion 2-categories will be discussed below in \S\ref{subsection:FBF2C}.

Nondegeneracy for braided fusion 2-categories is based on the notion of the sylleptic center for braided monoidal 2-categories, which categorifies the symmetric center of a braided monoidal 1-category.

\begin{defn}[\cite{Cr}]\label{def:syllepticcenter}
    The sylleptic center of a braided monoidal 2-category $\fB$, denoted by $\cZ_{(2)}(\fB)$, is the sylleptic monoidal 2-category defined as follows:
    \begin{itemize}
        \item Objects are given by pairs $(B,v_{B,-})$ where $B$ is an object of $\fB$ and $v_{B,-}$ is an invertible modification $v_{B||X}$ given on $X$ in $\mathfrak{B}$ by:
        \begin{equation}
            \begin{tikzcd}[sep=small]
                B\otimes X \arrow[rr, "\simeq", "\Downarrow v_{B||X}"'{yshift=-6pt}] \arrow[rdd,"R_{B|X}",swap]&  & B\otimes X.\\
               & & \\
                & X\otimes B \arrow[ruu,"R_{X|B}",swap]&
            \end{tikzcd}
        \end{equation}
        \item A 1-morphism between $(B,v_B)$ and $(C,v_{C})$ is a 1-morphism $f:B \to C$ in $\mathfrak{B}$ such that the following diagram commutes for every object $X$ of $\mathfrak{B}$:
         \begin{equation}
            \begin{tikzcd}[sep=tiny]
                B\otimes X \arrow[rr, "\Downarrow v_{B||X}"{yshift=-6pt},swap] \arrow[ddd]\arrow[rdd,"R_{B|X}" {xshift=4pt},swap]&  & B\otimes X \arrow[ddd]\\
               & {} & \\
                & X\otimes B \arrow[ruu,"R_{X|B}"{xshift=-2pt},swap] \arrow[ddd, "\Downarrow v_{C||X}"{yshift=-3pt}]& \\
                C\otimes X \arrow[rr] \arrow[rdd,"R_{C|X}",swap]& & C\otimes X. \\
                 & {} & \\
                 & X\otimes C \arrow[ruu,"R_{X|C}",swap]&
            \end{tikzcd}
        \end{equation}
        \item The 2-morphisms are exactly the 2-morphisms in $\mathfrak{B}$.
    \end{itemize}
\end{defn}

For our purposes, it is useful to note that it was argued in \cite[Proposition 5.1.8]{xu2024higher} that the sylleptic center of a braided fusion 2-category is always a fusion 2-category. We can now introduce the precise mathematical condition corresponding to the physical principle of remote detectability in (3+1)d. We will come across many variants of this notion.

\begin{defn}\label{def:degeneracy}
A braided fusion 2-category $\fB$ is \textit{nondegenerate} if $\cZ_{(2)}(\fB)\cong \tVect$.
\end{defn}

\subsection{Fermionic Topological Orders as Central Fusion 2-Categories}\label{subsection:perlimFTO}
To transition from bosonic theories to fermionic ones we will consider braided fusion 2-categories enriched in $\tSVect$. Physically, this corresponds to adding ``local fermions'' so that the theory  requires a spin structure to be defined.

\begin{defn}
    A $\mathbf{2SVect}$-central braided fusion 2-category is a braided fusion 2-category $\fB$ equipped with a sylleptic monoidal functor $\tSVect \to \cZ_{(2)}(\fB)$.\footnote{We will only be concerned with $\tSVect$-enrichments, but this definition of enrichment is reasonable if one replaces $\tSVect$ by any sylleptic fusion 2-category.}
\end{defn}

To get the fully fledged fermionic TQFT, we also need to introduce a notion of nondegeneracy for $\mathbf{2SVect}$-central braided fusion 2-categories. 

\begin{defn}
A $\mathbf{2SVect}$-central braided fusion 2-category is \textit{nondegenerate} if the sylleptic monoidal functor $\tSVect \to \cZ_{(2)}(\fB)$ is an equivalence.
\end{defn}

\begin{ex}
    As an illustration of the enrichment principle, we note that in (2+1)d, fermionic topological orders are nondegenerate $\SVect$-central braided fusion 1-categories, also presented as spin Chern-Simons theories \cite{Dijkgraaf:1989pz,jenquin2006spin,Seiberg:2016rsg}, see also \cite{Delmastro:2019vnj,Debray:2023iwf} exploring abelian spin Chern-Simons with global symmetry.
\end{ex}

By analogy with the 1-categorical description of (2+1)d fermionic topological orders, (3+1)d fermionic topological orders  are described in terms of nondegenerate $\mathbf{2SVect}$-central braided fusion 2-categories, and their classification was established in \cite{JF}. Before stating the formal classification theorem, we review the generalized 
cohomology theory $\SH$ known as (extended) supercohomology \cite{Wang:2017moj}.\footnote{There is also a notion of restricted supercohomomology given in \cite{Freed:2006mx,Gu:2012ib}, which only considers two of the homotopy group; we will not be using that definition of supercohomology.} The homotopy groups of the spectrum corresponding to $\SH$ are given by
\begin{equation}
    \pi_{-2}\SH = \Z/2 , \quad  \pi_{-1}\SH = \Z/2, \quad  \pi_{0}\SH = \mathbb C^\times\,,
\end{equation}
with Postnikov extension data given by
\begin{equation}\label{eq:postnikov}
    \Sq^2:\pi_{-2}\SH \rightarrow \pi_{-1}\SH, \quad (-1)^{\Sq^2}:\pi_{-1}\SH \rightarrow \pi_{0}\SH.
\end{equation}
Another way of realizing supercohomology is as a shift of the Picard space 
\begin{equation}\label{eq:2SVectgroups}
    (\tSVect)^\times = \Z/2\boldsymbol{\cdot}\rB\Z/2\boldsymbol{\cdot}\rB^2 \mathbb{C}^\times.
\end{equation}
Physically, the $\rB\Z/2$ represents the vacuum and the fermion line $f$, and the $\Z/2$ represents the condensation of $f$ on a surface.

Having briefly reviewed supercohomology, we can now state the classification result from \cite{JF}.

\begin{thm}{\cite[Corollary V.4]{JF}}\label{thm:fermionicTQFT}
    Fermionic topological orders in (3+1)d with nondegenerate local ground states are, canonically, gauge theories for finite groups $G$, with the Dijkgraaf–Witten action given by a class in $\SH^4(\rB G)$.
\end{thm}

Another feature of supercohomology  is that the  supercohomology of a space $X$ can be twisted by a class in $\rH^2(X;\Z/2)$, leading to twisted supercohomology. This will be relevant for the classification of (3+1)d topological orders with an emergent fermion as well as (3+1)d $G$-SETs. This twisting arises because $(\tSVect)^\times$ has a nontrivial space of automorphisms given by $\sAut^{\mathrm{br}}(\SVect) \simeq \rB \Z/2$ -- with nontrivial braided natural isomorphism given by the fermion parity operator $(-1)^F$. Given a space $X$ equipped with a map $\kappa:X\to  \rB^2\Z/2$, or, equivalently, a $\rB\Z/2$-equivariant structure, the $\kappa$-twisted $n$-th cohomology group of $X$ is the group $\SH^{n+\kappa}(X)$ of homotopy classes of $\rB\Z/2$-equivariant maps from $X$ to $\rB^{n-2}\tVect^\times$.

\begin{rem}\label{rem:loopscenter}
    While the classification of fermionic strongly fusion 2-categories requires twisted supercohomology \cite{Decoppet:2024htz}, this is not the case in the context of genuinely fermionic TQFTs. More precisely, let $\tSVect^\varpi_G$ be the fermionic strongly fusion 2-category classified by the finite group $G$, and classes $\tau\in\rH^2(\rB G;\mathbb{Z}/2)$, and $\varpi\in \SH^{4+\tau}(\rB G)$. The class $\tau$ corresponds to a central extension of the finite group $G$ by $\mathbb{Z}/2$. This is a super-group $(\widetilde{G},z)$, where $0\neq z\in\mathbb{Z}/2$. Now, we have by inspection that $$\Omega\cZ(\tSVect^\varpi_G)\cong \mathbf{Rep}(\widetilde{G},z).$$
    But it is well-known that there is a symmetric monoidal functor $\SVect\rightarrow \mathbf{Rep}(\widetilde{G},z)$ if and only if $(\widetilde{G},z) = G\times\mathbb{Z}/2$, or equivalently $\tau$ is trivial. It follows that there exists a braided monoidal functor $\tSVect\rightarrow \cZ(\tSVect^\varpi_G)$ if and only if $\tau$ is trivial, and thereby explains why there is no twist in the statement of Theorem \ref{thm:fermionicTQFT} above. To motivate this physically, observe that the data of the braided functor $\SVect\rightarrow \mathbf{Rep}(\widetilde{G},z)$ corresponds to selecting the local fermion.
\end{rem}

The previous remark  allows us to point out that the proof of \cite[Corollary V.4]{JF} is more precise:\ It shows that every nondegenerate $\mathbf{2SVect}$-central braided fusion 2-category arises as the centralizer, in the sense of \cite[Definition 4.1.1]{xu2024higher}, of a braided monoidal 2-functor
$$\tSVect\rightarrow \cZ(\tSVect^\varpi_G)\,$$
for some $\varpi \in \SH^4(\rB G)$.

\subsection{Crossed Braided Fusion 2-Categories}\label{subsection:extensions}

As emphasized in \S\ref{subsub:classifyingSET}, group-graded (higher) fusion categories play a central role in the study of symmetry-enriched topological orders. Accordingly, we briefly review the relevant mathematical framework -- extension theory -- describing the data necessary for constructing such graded (higher) fusion categories.

We start by reviewing the construction in (2+1)d, where coupling a (2+1)d topological order given by a nondegenerate braided fusion 1-category $\cB$ to a 0-form $G$-symmetry amounts to introducing topological $g$-defects implementing the $G$-symmetry. Transporting a quasiparticles given by an object in $\cB$ around one of these $g$-defects implements an action of $g \in G$, potentially permuting the topological charge of the quasiparticle.
Mathematically, this amounts to building a $G$-crossed braided fusion 1-category $\mathcal{A} = \bigoplus_{g\in G} \mathcal{A}_g$, whose trivially graded component is $\mathcal{A}_e=\mathcal{B}$. We also say that $\mathcal{A}$ is a $G$-crossed braided extension of $\mathcal{B}$. Physically speaking, the sectors labeled $\cA_g$ describe topologically distinct $g$ defects, which interact in a complicated fashion described, for instance, in \cite[Definition 8.24.1]{etingof2016tensor}.
Most of our results will be concerned with \textit{faithfully graded} extensions, meaning that the subcategory $\cA_g$ is nonzero for all $g\in G$.
Provided that the grading is faithful, it is a well-known result of \cite{ENO2} that $G$-crossed braided extensions of $\mathcal{B}$ are classified by homotopy classes of maps $\rB G\rightarrow \rB \Pic(\mathcal{B})$, from the classifying space of $G$ to the delooping of the space $\Pic(\mathcal{B})$ of invertible $\mathcal{B}$-module 1-categories. A different perspective on this classification was given in \cite[Section 8.2]{davydov2021braided}.

Before moving on to (3+1)d, it is useful to recast the above definition of a $G$-crossed braided fusion 1-category. In order to do so, we employ the notion, studied in \cite{D7}, of a \textit{rigid algebra} in a fusion 2-category. A rigid algebra in fusion 2-category is an algebra whose multiplication map has a right adjoint as a map of bimodules. This notion internalizes that of a multifusion 1-category. Namely, rigid algebras in $\mathbf{2Vect}$ are precisely multifusion 1-category. On the other hand, fusion 1-categories are recovered by considering \textit{strongly connected} rigid algebras in $\mathbf{2Vect}$, that is, rigid algebras for which the unit 1-morphism is the inclusion of a direct summand. This notion was introduced in \cite{JFR}. Slightly more generally, strongly connected rigid algebras in $\mathbf{2Vect}_G$ are (not necessarily faithfully) $G$-graded fusion 1-categories \cite[Example 5.2.4]{D4}. More directly relevant to our present purposes are strongly connected rigid algebras in $\mathcal{Z}(\mathbf{2Vect}_G)$, which recover the $G$-crossed fusion 1-categories introduced in \cite{Tur}. Finally, taking the braiding into account, we have that (not necessarily faithfully graded) $G$-crossed braided fusion 1-categories are precisely strongly connected braided rigid algebras in $\mathcal{Z}(\mathbf{2Vect}_G)$. The obvious advantage of this definition is that it can be made regardless of the ambient dimension.

In (3+1)d, topological orders with $G$-symmetry are described by $G$-crossed braided fusion 2-categories. The above discussion naturally leads to the following definition.

\begin{defn}\label{def:GcrossedBF2C}
A (not necessarily faithfully graded) $G$-crossed braided fusion 2-category is a strongly connected braided rigid algebra in $\cZ(\mathbf{3Vect}_G)$.
\end{defn}

Inspired by the classification of faithfully graded $G$-crossed braided fusion 1-categories obtained in \cite{ENO2}, we exhibit a homotopy theoretic classification of faithfully graded $G$-crossed braided fusion 2-categories. In order to do so, we consider the Picard space $\sPic(\fB)$ of a braided fusion 2-category $\mathfrak{B}$. By definition, the space $$\sPic(\fB):=\mathbf{Mod}(\mathfrak{B})^{\times}$$ consists of invertible $\fB$-module 2-categories and invertible $\fB$-module morphisms. In particular, it is a group-like topological monoid under the relative tensor product, and may therefore be delooped.

\begin{thm}\label{thm:crossedextn}
Let $\fB$ be a braided fusion 2-category. Faithfully graded $G$-crossed braided extensions of $\fB$ are classified by homotopy classes of maps $\rB G \rightarrow \rB \sPic(\fB)$.
\end{thm}
\begin{proof}
Recall from \cite[Theorem 3.11]{D11} that, given a fusion 2-category $\mathfrak{C}$, faithfully $G$-graded extensions $\mathfrak{D}=\boxplus_{g\in G}\mathfrak{D}_g$ of $\mathfrak{C}=\mathfrak{D}_e$ are classified by homotopy classes of maps $\rB G\rightarrow \rB \sBrPic(\mathfrak{C})$ into the delooping of the space $\sBrPic(\mathfrak{C})$ of invertible $\mathfrak{C}$-$\mathfrak{C}$-bimodule 2-categories.

Let now $\mathfrak{A}$ be a faithfully graded $G$-crossed braided fusion 2-category.
Thanks to the forget monoidal 3-functor $\mathcal{Z}(\mathbf{3Vect}_G)\rightarrow\mathbf{3Vect}_G$, we see that $\mathfrak{A}=\boxplus_{g\in G}\mathfrak{A}_g$ is in particular a $G$-graded fusion 2-category.
We shall write $\mathfrak{B} = \mathfrak{A}_e$ for the braided fusion 2-category corresponding to the trivially graded component of $\mathfrak{A}$.
By assumption, the underlying fusion 2-category of $\mathfrak{A}$ is a faithfully $G$-graded extension of $\mathfrak{B}$, so that it corresponds to a map $\rB G\rightarrow \rB \sBrPic(\mathfrak{B})$.
But, as $\mathfrak{A}$ is a braided algebra in $\cZ(\mathbf{3Vect}_G)$, the map $\rB G\rightarrow \rB \sBrPic(\mathfrak{B})$ must factor through $\rB G\rightarrow \rB \sPic(\mathfrak{B})$.
Namely, the crossed braiding on $\mathfrak{A}$ supplies a natural 2-equivalence
$$c_{A,B}:A\otimes B \xrightarrow{\simeq} B\otimes A$$
for $A$ in $\mathfrak{A}$ and $B$ in $\mathfrak{B}=\mathfrak{A}_e$.
As $c_{A,B}$ is monoidal in $A$, it follows that each graded component $\mathfrak{A}_g$, which is a priori an invertible $\mathfrak{B}$-$\mathfrak{B}$-bimodule 2-category, is actually an invertible $\mathfrak{B}$-module 2-category. Moreover, as $c_{A,B}$ is monoidal in $B$, we indeed have the claimed lift $\rB G\rightarrow \rB \sPic(\mathfrak{B})$.

Let us now fix a braided fusion 2-category $\mathfrak{B}$.
We explain how to construct a faithfully graded $G$-crossed braided fusion 2-category from a map $\rB G\rightarrow \rB \sPic(\mathfrak{B})$.
Firstly, by considering the composite $\rB G\rightarrow \rB \sPic(\mathfrak{B})\rightarrow \rB \sBrPic(\mathfrak{B})$, we obtain a faithfully $G$-graded extension $\mathfrak{A}=\boxplus_{g\in G}\mathfrak{A}_g$ of $\mathfrak{B}$.
Let $g,h\in G$. Given that $\mathfrak{A}_g$ is an invertible $\mathfrak{B}$-module 2-category, we have equivalences of $\mathfrak{B}$-module 2-categories:
$$
\begin{tabular}{c c c}
$\mathfrak{A}_h$ & $\xrightarrow{\simeq}$ & $\mathbf{Fun}_{\mathfrak{B}}(\mathfrak{A}_g,\mathfrak{A}_{gh})\,,$\\
$A$ & $\mapsto$ & $(-)\otimes A$
\end{tabular}
\quad
\begin{tabular}{c c c}
$\mathfrak{A}_{ghg^{-1}}$ & $\xrightarrow{\simeq}$ & $\mathbf{Fun}_{\mathfrak{B}}(\mathfrak{A}_g,\mathfrak{A}_{gh})\,.$\\
$A$ & $\mapsto$ & $A\otimes (-)$
\end{tabular}
$$
Combining these equivalences together yields a $G$-action on $\mathfrak{A}$ that acts by conjugation on the grading.
Moreover, by construction, there is a coherent natural 2-equivalence
$$c_{A,B}:A\otimes B \xrightarrow{\simeq} g(B)\otimes A$$
for every $A$ in $\mathfrak{A}_g$ and $B$ in $\mathfrak{A}$. This provides $\mathfrak{A}$ with a $G$-crossed braiding. In summary, we have shown that $\mathfrak{A}$ is a faithfully graded $G$-crossed braided extension of $\mathfrak{B}$.

Finally, one checks that these two operations are inverse to each other.
\end{proof}

\subsection{Braided Extensions}\label{subsection:braidedext}

In \cite[Section 4]{davydov2021braided}, a different kind of extension theory tailored to braided fusion 1-categories was developed. More precisely, the authors classify braided extensions of braided fusion 1-categories by way of the braided Picard space. Our classification of fermionic braided fusion 2-categories in \S\ref{subsection:FBF2C} will employ a 2-categorical version of this construction, which utilizes the braided Picard space
\begin{equation}\label{eq:picbrequiv}
    \sPic^{\mathrm{br}}(\fB) := \cZ(\Mod (\fB))^\times. 
\end{equation}
This is a braided group-like topological monoid, so that it can be delooped twice.

Our next result is an analogue of \cite[Theorem 3.11]{D11} for braided extensions.

\begin{prop}\label{prop:braidedextn}
    Let $\fB$ be a braided fusion 2-category, and $E$ a finite abelian group. Faithfully $E$-graded braided extensions of $\mathfrak{B}$ are classified by homotopy classes of maps $\rB^2 E \rightarrow \rB^2 \sPic^{\mathrm{br}}(\mathfrak{B})$.
\end{prop}
\begin{proof}
We begin by making two formal observations. Firstly, it follows from the definitions that braided rigid algebras in $\cZ(\Mod (\fB))$ are precisely braided multifusion 2-categories $\mathfrak{A}$ equipped with a braided monoidal functor $\mathfrak{B}\rightarrow \mathfrak{A}$. Secondly, an $E$-graded braided algebra in $\cZ(\Mod (\fB))$, or, more generally, any braided monoidal higher category, is a lax braided monoidal functor $E\rightarrow \cZ(\Mod (\fB))$.

Given a map of spaces $\rB^2 E \rightarrow \rB^2 \sPic^{\mathrm{br}}(\mathfrak{B})$. We can equivalently consider the braided monoidal 2-functor $E\rightarrow \sPic^{\mathrm{br}}(\mathfrak{B})$. In particular, we get a braided monoidal 2-functor $E \rightarrow \cZ(\Mod (\fB))$. After linearizing and Cauchy completing the source, we therefore obtain a braided monoidal 2-functor $F:\mathbf{3Vect}_E \rightarrow \cZ(\Mod (\fB))$. But there is a canonical braided rigid algebra $\mathbf{2Vect}(E)$ in $\mathbf{3Vect}_E$. Its image under $F$ is therefore an $E$-graded braided rigid algebra in $\cZ(\Mod (\fB))$. This is the desired faithfully $E$-graded braided extension of $\mathfrak{B}$.

Conversely, given any $E$-graded braided extension $\mathfrak{A}$ of $\mathfrak{B}$, we obtain a lax braided monoidal functor $E \rightarrow \cZ(\Mod (\fB))$. We claim that this functor must factor through $\sPic^{\mathrm{br}}(\mathfrak{B})$ and is in fact (strong) braided monoidal. Both of these follow from \cite[Theorem 3.11]{D11} using the fact that the composite monoidal 2-functor $$\cZ(\Mod (\fB))\rightarrow \Mod (\fB)\rightarrow \mathbf{Bimod}(\fB)$$ reflects invertibility of objects as well as invertibility of morphisms.
\end{proof}

\subsection{The Wang-Wen-Witten Construction}\label{subsection:WWW}
To put the classification data for (3+1)d $G$-SETs in context, we now review the symmetry extension construction of \cite{Wang:2017loc}. This is a procedure used to construct a $d$-dimensional bosonic $G$-symmetric TQFT that saturates a particular value of the anomaly in $\rH^{d+1}(\rB G;\mathbb C^\times)$. The theory is built from the following set of data:

\begin{ansatz}[Wang-Wen-Witten]\label{ansatz:WWW}
    An $n$-dimensional topological theory with $G$-symmetry and anomaly valued in $\rH^{n+1}(\rB G; \mathbb{C}^\times)$ can be realized by a $K$-gauge theory, using the following set of data:
    \begin{itemize}
          \item A short exact sequence $1 \rightarrow K \rightarrow H \rightarrow G\rightarrow 1$, where the normal subgroup $K$ is abelian.
          \item A class $\rH^n(\rB H; \mathbb{C}^\times)$ parametrizing an $H$-invertible TQFT on the boundary that realizes the $G$-symmetry.
          \item A class in $\rH^{n}(\rB K; \mathbb{C}^\times)$ giving the Dijkgraaf-Witten action for the $K$-gauge theory with $G$-symmetry, that realizes the anomaly in the bulk. 
    \end{itemize}
\end{ansatz}
Presented in this way, the data in the ansatz has a natural generalization to the fermionic case, which we presented in \Cref{ansatz:fernionicTQFT}.
In the next proposition, we give the conditions that the class in $\rH^{n+1}(\rB G;\mathbb C^\times)$  must satisfy in order to actually construct the $K$-gauge theory.

\begin{prop}[\cite{Wang:2017loc}]
Let $G$ be a finite group, and $\omega \in \rH^{d+1}(\rB G; \mathbb C^\times)$. Let also
\begin{equation}
    1 \longrightarrow K \longrightarrow H \xlongrightarrow{p} G \longrightarrow 1
\end{equation}
be the central extension of $G$ by an abelian group $K$ parametrized by the class $e\in \rH^2(\rB G; K)$. Suppose that $\omega = e \cup z$ for $z \in \rH^{d-1}(\rB G; \widehat{K})$, and $p^* \omega =[0] \in \rH^{d+1}(\rB H;\mathbb C^\times)$. Then one can construct a $K$-gauge theory, with Dijkgraaf-Witten action given by a class in $\rH^d(\rB K; \mathbb{C}^\times)$, and $G$-symmetry that realizes the $\omega$ anomaly. 
\end{prop}
\begin{proofsketch}
Consider a bulk $(d+1)$-dimensional manifold $N$ with $\partial N = M$, and $g: N \rightarrow \rB G$ which lifts to $h: M\rightarrow \rB H$ on the boundary.
If  $p^* \omega$ is trivial in $\rH^{d+1}(\rB H;\mathbb C^\times)$, then $p^* \omega = d \pi$ for some cochain $\pi \in \mathrm C^d(\rB H;\mathbb C^\times)$. The boundary partition function is 
\begin{equation}\label{eq:boundaryaction}
   Z \sim \sum_{h\in [M,\rB H]}  \exp\left(-2 \pi i \int_M h^*\pi \right),
\end{equation}
and provides a trivial gapped boundary theory with symmetry $H$ and anomaly $\omega$. In particular, we see that $\rH^{d}(\rB H; \mathbb{C}^\times)$ parametrizes the symmetry preserving gapped boundaries. We abuse notation and let $\omega\in \mathrm{Z}^{d+1}(\rB G; \mathbb{C}^\times) $ also denote a cocycle representative of the anomaly. We choose $K$ so that $\omega = e \cup z$ for $e \in \mathrm{Z}^{2}(\rB G ; K)$ and $z \in \mathrm{Z}^{d-1}(\rB G ; \widehat{K})$. Such a finite abelian group $K$ always exists thanks to \cite{Tachikawa:2017gyf}. We take $a \in \mathrm{C}^1(\rB G;K)$ such that $d a = e$, and hence $\pi = p^*(- a \cup z)$. 
By substituting $- a \cup z$ into the exponent in \Cref{eq:boundaryaction} we get 
\begin{equation}
    \exp\left( 2\pi i \int_M h^* a \cup h^* z \right), 
\end{equation}
where we have implicitly identified $a$ with its pullback to a cocycle on $\rB H$.
By restricting $h^*$ to the part that pulls back cocycles valued in $\rB G$, and treating $z$ as a cocycle over $\rB G$, we get that $h^* z = g^*z$. Letting $a' = h^*a  \in \mathrm{C}^1(M; K)$ so that $d a' = g^* e$, and adding in the boundary term $d (a' \cup b)$ into the exponential, the full partition function can be written as 
\begin{equation}\label{eq:full}
   Z \sim \sum_{\substack{a' \in \mathrm{C}^1(M;K) \\ b\in \mathrm{C}^{d-2}(M;\widehat{K})}}  \exp\left(2 \pi i \int_M a' \cup g^*z +g^*e \cup b + a'\cup db \right),
\end{equation}
where the coupling to $G$-symmetry is through the terms $g^*z$ and $g^*e$.
\end{proofsketch}

\section{The Classification of Nondegenerate Braided Fusion 2-Categories}\label{section:BF2C}

We present a classification of nondegenerate braided fusion 2-categories different from that obtained in \cite{JF}.
More precisely, employing a 2-categorical version of the (de-)equivariantization correspondence given in \cite{DGNO}, we reduce the classification of nondegenerate braided fusion 2-categories to that of faithfully graded crossed braided fusion 2-categories.
In fact, the classification reduces to describing crossed braided extensions of braided strongly fusion 2-categories, which we carry out using results from the previous section.

\subsection{(De-)Equivariantization}

Our next result is a 2-categorical version of the (de-)equivariant\-ization correspondence of \cite[Theorem 4.44]{DGNO}, which posits an equivalence between the 2-category of braided fusion 1-categories containing $\mathbf{Rep}(G)$ and the 2-category of $G$-crossed braided fusion 1-categories. As explained in \cite[Remark 3.10]{D9}, this correspondence admits a streamlined proof by unpacking the notion of a strongly connected rigid algebra on both sides of the equivalence of braided fusion 2-categories 
\begin{equation}\label{eq:Z2VGZ2RG}
\mathcal{Z}(\mathbf{2Vect}_G)\simeq\mathcal{Z}(\mathbf{2Rep}(G)).
\end{equation}
We use a categorification of this argument to argue that (de-)equivariantization holds for braided fusion 2-categories.

\begin{thm}{(2-categorical (de-)equivariantization)}\label{thm:deequiv}
There is an equivalence of 3-categories between braided fusion 2-categories containing $\mathbf{2Rep}(G)$ fully faithfully and $G$-crossed braided fusion 2-categories.
\end{thm}
\begin{proof}
By definition, we have that $G$-crossed braided multifusion 2-categories are rigid braided algebras in $\mathcal{Z}(\mathbf{3Vect}_G)$. On the other hand, we have $\mathbf{Mod}(\mathbf{2Rep}(G))\simeq \mathbf{3Rep}(G)$ as (symmetric) fusion 3-categories as the right hand-side is connected. Moreover, rigid braided algebras in $\mathcal{Z}(\mathbf{Mod}(\mathbf{2Rep}(G)))$ are braided multifusion 2-categories equipped with a braided 2-functor from $\mathbf{2Rep}(G)$. Now, we have that the fusion 3-categories $\mathbf{3Vect}_G$ and $\mathbf{3Rep}(G)$ are Morita equivalent via $\mathbf{3Vect}$. It follows that there is an equivalence of braided fusion 3-categories
\begin{equation}\label{eq:Z3VGZ3RG}
\mathcal{Z}(\mathbf{3Vect}_G)\simeq\mathcal{Z}(\mathbf{3Rep}(G)).
\end{equation}
This is a consequence of the fact that Morita equivalent fusion higher categories have equivalent Drinfeld centers.
This establishes that the 3-category of $G$-crossed braided multifusion 2-categories is equivalent to the 3-category of braided multifusion 2-categories equipped with a braided 2-functor from $\mathbf{2Rep}(G)$.
Namely, braided rigid algebras are preserved by arbitrary braided monoidal functors.

In order to conclude the proof, it remains to show that this correspondence sends $G$-crossed braided fusion 2-categories to braided fusion 2-categories equipped with a fully faithful braided 2-functor from $\mathbf{2Rep}(G)$. We have recalled above that an algebra in a higher fusion category is called strongly connected if its unit 1-morphism is the inclusion of a summand. It is clear that strongly connected rigid algebras in $\mathcal{Z}(\mathbf{3Vect}_G)$ are $G$-crossed braided fusion 2-categories. On the other hand, strongly connected rigid algebras in $\mathcal{Z}(\mathbf{Mod}(\mathbf{2Rep}(G)))$ are braided fusion 2-categories equipped with a fully faithful braided 2-functor from $\mathbf{2Rep}(G)$. The proof is completed by appealing to Equation \eqref{eq:Z3VGZ3RG}.
\end{proof}

We are specifically interested in nondegenerate braided fusion 2-categories and their behavior under (de-)equivariantization. To this end, we will consider a 2-categorical version of \cite[Proposition 4.56]{DGNO}.

\begin{prop}\label{prop:nondegdeequiv}
There is an equivalence of 3-categories between nondegenerate braided fusion 2-categories containing $\mathbf{2Rep}(G)$ and faithfully graded $G$-crossed braided fusion 2-categories whose trivially graded component is nondegenerate.
\end{prop}
\begin{proof}
It will be enough to prove that the $G$-equivariantization $\mathfrak{B}^G$ of a $G$-crossed braided fusion 2-category $\mathfrak{B} = \boxplus_{g\in G}\mathfrak{B}_g$ is nondegenerate if and only if it is both faithfully graded and its trivially graded component $\mathfrak{B}_e$ is nondegenerate. Firstly, let $N\trianglelefteq G$ be the normal subgroup of $G$ on which the grading of $\mathfrak{B}$ is supported. The braided fusion 2-category $\mathfrak{B}^G$ contains a canonical copy of $\mathbf{2Rep}(G)$ as the $G$-equivariantization of the canonical copy of $\mathbf{2Vect}$ in $\mathfrak{B}$. But it follows from the definition of a $G$-crossed braiding that a loop in $\mathbf{2Rep}(G)$, i.e.\ an object of $\mathbf{Rep}(G)=\Omega\mathbf{2Rep}(G)$, can be lifted to $\mathcal{Z}_{(2)}(\mathfrak{B}^G)$ if and only if the subgroup $N$ acts trivially on it.
Thence, there is an embedding $\mathbf{Rep}(G/N)\hookrightarrow \Omega\mathcal{Z}_{(2)}(\mathfrak{B}^G)$.
It is therefore enough to consider the case when the grading on $\mathfrak{B}$ is faithful, that is, $G=N$. But, under this assumption, we have $\mathcal{Z}_{(2)}(\mathfrak{B}^G)\simeq \mathcal{Z}_{(2)}(\mathfrak{B}_e)$, from which the result follows.
\end{proof}

\begin{rem}
We now provide physical intuition behind the above results. One should think of $G$-crossed braided fusion 2-categories as symmetries of (3+1)d quantum field theories where the codimension 1 operators implement the $G$ 0-form symmetry. In the case of a $G$-crossed braided fusion 2-categorical symmetry, 0-form symmetry operators arise as condensation operators of a 1-form symmetry. Therefore, gauging the 0-form $G$-symmetry, amounts to gauging  the ``genuine'' codimension $2$ operators, i.e.\ surface operators implementing a categorical 1-form symmetry. One then expects the dual symmetry to be a categorical 1-form symmetry, given by $\tRep(G)$.
\end{rem}

\subsection{Nondegenerate Braided Bosonic Fusion 2-Categories and TQFTs}\label{subsection:BBF2C}

We first classify faithfully $G$-crossed braided bosonic strongly fusion 2-categories. Appealing to equivariantization then yields the classification of nondegenerate braided bosonic fusion 2-categories. As a necessary first step, we must classify braided bosonic strongly fusion 2-categories. Thanks to \Cref{prop:braidedextn} applied to the braided fusion 2-category $\mathbf{2Vect}$, we obtain the following result.

\begin{lem}
Braided bosonic strongly fusion 2-categories are classified by a finite abelian group $E$ together with a class $\beta$ in $\rH^5(\rB^2 E;\mathbb C^\times)$.\footnote{More precisely, this statement gives the classification up to graded braided equivalence. The classification up to braided equivalence can be deduced by considering automorphisms of the grading group $E$.}
\end{lem}

We write $\tVect^\beta_{E}$ for the braided bosonic strongly fusion 2-categories corresponding to the above data. We must now investigate the structure of the Picard space of such braided fusion 2-categories.

\begin{lem}\label{lem:bosonicpic}
    The homotopy groups of $\sPic^{}(\tVect^\beta_{E})$ are given by:
\begin{equation}\label{eq:pic2Vec}
\begin{tabular}{|c|c|c|c|c|}
\hline
$\pi_0$ & $\pi_1$ & $\pi_2$ & $\pi_3$\\
\hline \\[-1em]
$*$ & $E$ &  $0$ &  $\mathbb{C}^{\times}$\\
\hline
\end{tabular}
\end{equation}
Moreover, the single nontrivial $k$-invariant is given by $\beta$.
\end{lem}
\begin{proof}
The groups $\pi_{\geq 1}(\sPic^{}(\tVect^\beta_{E}))$ as well as the $k$-invariant can be read off directly from the braided bosonic strongly fusion 2-category $\tVect^\beta_{{E}}$. Namely, by construction,
there is a fiber sequence 
\begin{equation}
    \rB \sPic(\tVect^\beta_{E}) \longrightarrow \rB^2 E \xlongrightarrow{\beta} \rB^5\mathbb C^\times\,.
\end{equation}
Thence, we only have to argue that $\pi_{0}(\sPic^{}(\tVect^\beta_{ E}))$ is trivial. But it was shown in \cite[Proposition 3.4.1]{Decoppet:2023bay} that every invertible $\tVect^\beta_{ E}$-bimodule 2-category is given by an autoequivalence of $\tVect^\beta_{ E}$. It follows immediately that every invertible $\tVect^\beta_{ E}$-module 2-category must be trivial as asserted.
\end{proof}

\Cref{thm:crossedextn} now immediately affords us a classification of faithfully graded $G$-crossed braided bosonic strongly fusion 2-categories. Namely, it follows from obstruction theory that homotopy classes of maps $\rB G\rightarrow \rB \sPic(\tVect^\beta_{E})$ correspond exactly to a homotopy class of maps $\tau: \rB G \to \rB^2 E$ together with a trivialization of the composite $\beta \circ \tau$. Recalling that the data of a trivialization of $\beta \circ \tau $ forms a torsor over $\rH^4(\rB G;\mathbb C^\times)$, we obtain the following classification result.

\begin{prop}\label{thm:AllBoson}
Fix $G$ a finite group. Faithfully graded $G$-crossed braided bosonic strongly fusion 2-categories are classified by a finite abelian group $E$, a class $\beta \in \rH^5(\rB^2 E;\mathbb C^\times )$, a class $\tau \in \rH^2(\rB G;E)$ such that $\beta \circ \tau$ is trivial, and a class $\pi \in \rH^4(\rB G;\mathbb C^\times)$.
\end{prop}

Given that the only nondegenerate braided bosonic strongly fusion 2-category is $\mathbf{2Vect}$, using \Cref{prop:nondegdeequiv}, we recover the well-known classification of nondegenerate braided bosonic fusion 2-categories, i.e.\ (3+1)d bosonic TOs. 

\begin{thm}[\cite{Lan_2018,JF}]\label{BosonicNondegen}
    Nondegenerate bosonic braided fusion 2-categories are classified by a finite group $G$ and a class $\pi \in \rH^4(\rB G; \mathbb{C}^\times)$.
\end{thm}

\subsection{Nondegenerate Braided Fermionic Fusion 2-Categories and TQFTs}\label{subsection:FBF2C}
We first classify braided fermionic strongly fusion 2-categories, and then compute their Picard spaces in order to apply extension theory.

\begin{lem}\label{lem:piPicbr}
The homotopy groups of $\sPic^{\br}(\tSVect)$ are given by
\begin{equation}\label{eq:picbr2svect}
\begin{tabular}{|c|c|c|c|c|}
\hline
$\pi_0$ & $\pi_1$ & $\pi_2$ & $\pi_3$\\
\hline \\[-1em]
$\Z/2$ & $\Z/2$ &  $\Z/2$ &  $\mathbb{C}^{\times}$\\
\hline
\end{tabular}
\end{equation}
\end{lem}
\begin{proof}
By way of \Cref{eq:picbrequiv}, it suffices to compute the homotopy groups of $\cZ(\mathbf{3SVect})^\times$.  We first consider $\Omega \cZ(\mathbf{3SVect})\cong \cZ_{(2)}(\tSVect)\cong \tSVect$, where the first equivalence is in \cite[Section IV.B]{JF} and the second holds by inspection.
It therefore only remains to argue that $\pi_0(\cZ(\mathbf{3SVect})^\times)\cong \Z/2$.
To this end, observe that $\sPic(\tSVect) = 0$.
Namely, equivalence classes of invertible $\mathbf{2SVect}$-module 2-categories correspond to Morita equivalence classes of invertible $\mathbf{SVect}$-central fusion 1-categories, of which there is only one.
Moreover, given that $\mathbf{2SVect}$ is symmetric, the data of a braiding on an invertible $\mathbf{2SVect}$-module 2-category amounts precisely to that of a braided monoidal autoequivalence of the identity 2-functor $\mathrm{Id}_{\mathbf{2SVect}}$ on $\mathbf{2SVect}$.
But we have that $\sAut^{\br}(\mathrm{Id}_{\mathbf{2SVect}})\cong\sAut^{\mathrm{sym}}(\mathrm{Id}_{\mathbf{SVect}})\cong\mathbb{Z}/2$, from which it follows that $\pi_0(\cZ(\mathbf{3SVect})^\times)\cong \Z/2$. This concludes the proof.

\end{proof}

It is important to remark that by nondegeneracy, the group $\pi_0(\sPic^{\br}(\tSVect))\simeq \Z/2$ acts nontrivially on $\tSVect^\times$. More precisely, the fiber sequence 
 \begin{equation}\label{eq:actionSH5}
     \begin{tikzcd}
         \SH^5 \arrow[r]& \mathrm{B}^2\sPic^{\br}(\tSVect) \arrow[d]\\
         & \rB^2 \Z/2
     \end{tikzcd}
 \end{equation}
coincides with the fiber sequence classifying the action of $\mathrm{B}\mathbb{Z}/2$ on $\SH^5$.

By \Cref{prop:braidedextn}, homotopy classes of maps $\rB^2E \rightarrow \rB^2 \sPic^{\br}(\tSVect)$ parametrize braided fermionic strongly fusion 2-categories. As explained in \cite[Section 4.2]{D11}, it is more convenient to re-express this data as follows:\ First, we have a map $\kappa:\rB^2 E \rightarrow \pi_0(\sPic^{\br}(\tSVect))$, inducing an action of $\rB^2 E$ on $\pi_{\geq 1}(\sPic^{\br}(\tSVect))\simeq \SH^5$. Secondly, we have a map $\rB^2 E \rightarrow \SH^5$ compatible with $\kappa$. This is precisely the data of a class in $\SH^{5+\kappa}(\rB^2 E)$. Hence we see that:

\begin{cor}\label{cor:BFSF2C}
Braided fermionic strongly fusion 2-categories are  classified by a finite abelian group $E$, a class $\kappa \in \rH^2(\rB^2 E;\Z/2)$, and a class $\varsigma  \in \SH^{5+\kappa}(\rB^2 E)$.
\end{cor}

\begin{rem}\label{rem:SandT}
In the case when $E = \Z/2$ and $\kappa$ is trivial, we have that $\SH^{5}(\rB^2 E) = 0$. The corresponding braided fermionic strongly fusion 2-category is $\mathbf{2SVect}_{\mathbb{Z}/2}$.
On the other hand, when $\kappa$ is the nontrivial class in $\rH^2(\rB^2 \mathbb{Z}/2;\Z/2)\cong\mathbb{Z}/2$, we have $\SH^{5+\kappa}(\rB^2 E)\cong\mathbb{Z}/2$. These two classes correspond to the nondegenerate braided fermionic fusion 2-categories denoted $\mathcal S = \mathcal{Z}(\mathbf{2SVect})$ and $\mathcal T$ in \cite{JF2,JFR}. Additionally, it follows from the above result that there are isomorphisms of groups $\mathrm{Aut}^{\br}(\mathcal{S})\cong\mathbb{Z}/16$ and $\mathrm{Aut}^{\br}(\mathcal{T})\cong\mathbb{Z}/16$. This was first observed in \cite[Equation 19]{JF2}.
Using our methods, this follow from the well-known fact that $\SH^{4+\kappa}(\rB^2 E)\cong\mathbb{Z}/16$.
Namely, the spaces $\sAut^{\br}(\mathcal{S})$ and $\sAut^{\br}(\mathcal{T})$ are both equivalent to the delooping of the space of maps from $\mathrm{B}^2E$ to $\mathrm{SH}^5$ with action classified by $\kappa$. We note that this identification is implicitly using the fact that the grading group $E$ has no automorphisms.
\end{rem}

We write $\tSVect^{(\kappa,\varsigma)}_{E}$ for the braided fermionic strongly fusion 2-category corresponding to the data in the corollary above.
In order to classify crossed braided extensions of $\tSVect^{(\kappa,\varsigma)}_{E}$, we presently compute the fundamental group of the corresponding Picard space.

\begin{lem}\label{lem:PicBFSF2C}
The group $\pi_0(\sPic^{}(\tSVect^{(\kappa,\varsigma)}_{E}))$ is trivial.
\end{lem}
\begin{proof}
Said differently, we must prove that every invertible $\tSVect^{(\kappa,\varsigma)}_{E}$-module 2-category is trivial. In order to do so, just as in the bosonic case treated above in \Cref{lem:bosonicpic}, it will suffice to show that every invertible $\tSVect^{(\kappa,\varsigma)}_{E}$-bimodule 2-category is quasi-trivial, i.e.\ arises from an autoequivalence of the fusion 2-category $\tSVect^{(\kappa,\varsigma)}_{E}$. 
So as to see this, note that it follows from \cite[Lemma 4.2.3]{D9} that every invertible $\tSVect^{(\kappa,\varsigma)}_{E}$-bimodule 2-category arises as a graded component in a graded extension of $\tSVect^{(\kappa,\varsigma)}_{E}$. But it follows from \cite{JFY}, that every such invertible $\tSVect^{(\kappa,\varsigma)}_{E}$-bimodule 2-category is quasi-trivial given that the corresponding graded component must contain an invertible object. This concludes the proof. 
\end{proof}

\begin{prop}\label{thm:EF}
   For $G$ a finite group. Faithfully graded $G$-crossed braided fermionic strongly fusion 2-categories are classified by a finite abelian group $E$, a class $\kappa \in \rH^2(\rB^2 E;\Z/2)$, a class $\varsigma \in \SH^{5+\kappa}(\rB^2 E)$, a class $\tau\in \rH^2(\rB G;E)$, such that $\varsigma \circ \tau$ is trivial in $\SH^{5+\kappa\circ\tau}(\rB G)$, and a class $\varpi \in \SH^{4+\kappa\circ\tau}(\rB G)$.
\end{prop}

\begin{proof}
By \Cref{cor:BFSF2C}, braided fermionic strongly fusion 2-categories are all of the form $\tSVect^{(\kappa,\varsigma)}_{E}$.
By \Cref{thm:crossedextn}, faithfully graded $G$-crossed braided extensions of $\tSVect^{(\kappa,\varsigma)}_{E}$ are classified by homotopy classes of maps $\mathrm{B}G\rightarrow \rB \sPic(\tSVect^{(\kappa,\varsigma)}_{E})$. In order to unpack the data of such maps, observe that the space $\rB\sPic(\tSVect^{(\kappa,\varsigma)}_{E})$ fits in the following fiber sequence
\begin{equation}\label{eq:pic2svectpullback}
\rB\sPic(\tSVect^{(\kappa,\varsigma)}_{E})\rightarrow\mathrm{B}^2E\xrightarrow{(\kappa,\varsigma)}\rB^2\sPic^{\br}(\tSVect).
\end{equation}
In particular, it follows that homotopy classes of maps $\rB G\rightarrow \rB\sPic(\tSVect^{(\kappa,\varsigma)}_{E})$ correspond to homotopy classes of maps $\tau: \rB G \rightarrow \rB^2 E$ together with a trivialization of the composite $\varsigma \circ \tau$. But, such trivializations form a torsor over the group $\SH^{4+\kappa\circ\tau}(\rB G)$. We may therefore reformulate this last piece of data as the requirement that $\varsigma \circ \tau$ be trivial in $\SH^{5+\kappa\circ\tau}(\rB G)$ together with the data of a class $\varpi \in \SH^{4+\kappa\circ\tau}(\rB G)$.
\end{proof}

It follows from an $S$-matrix argument, that is, using \cite[Theorem 2.57]{JFR}, that every nondegenerate braided fermionic strongly fusion 2-category $\mathfrak{B}$ must satisfy $\pi_0(\mathfrak{B})\cong\mathbb{Z}/2$. But all such braided fermionic braided strongly fusion 2-categories are described in \Cref{rem:SandT}.
In particular, the only nondegenerate braided fermionic strongly fusion 2-categories are $\mathcal S = \mathcal{Z}(\mathbf{2SVect})$ and $\mathcal T$. Using \Cref{prop:nondegdeequiv}, we obtain a classification of nondegenerate braided fermionic fusion 2-categories, i.e.\ (3+1)d fermionic TOs, that ought to be compared with \cite[Corollary V.5]{JF}.

\begin{thm}\label{FermionicNondegen}
   Nondegenerate fermionic braided fusion 2-categories are classified by a finite group $G$, a class $\varsigma \in \SH^{5+\kappa}(\rB^2 \Z/2)\cong\mathbb{Z}/2$, where $\kappa$ is the nontrivial class in $\rH^2(\rB^2 \Z/2;\Z/2)$, a class $\tau\in \rH^2(\rB G;\Z/2)$, such that $\varsigma \circ \tau$ is trivial in $\SH^{5+\tau}(\rB G)$, and a class $\varpi \in \SH^{4+\tau}(\rB G)$.
\end{thm}

\begin{rem}\label{rem:fermioniccenters}
Some, but not all, of the nondegenerate fermionic braided fusion 2-categories in \Cref{FermionicNondegen} arise as the Drinfeld center of a fusion 2-category. For instance, the braided fusion 2-category $\mathcal{T}$ discussed in \Cref{rem:SandT} does not by \cite{JFR,D9}. More generally, we expect that a nondegenerate fermionic braided fusion 2-category arises as a Drinfeld center if and only if the class $\varsigma$ is trivial.
Very explicitly, it would suffice to check that the Drinfeld center of the fermionic strongly fusion 2-category classified by $\tau\in \rH^2(\rB G;\Z/2)$ and $\varpi \in \SH^{4+\tau}(\rB G)$ as in \cite[Theorem 4.5]{D11} is the nondegenerate braided fermionic fusion 2-category classified by $\tau$ and $\varpi$.
\end{rem}

\begin{rem}\label{rem:redundancies}
The equivalence relation that has to be imposed on the data given in \Cref{FermionicNondegen} is quite subtle.
Firstly, it is well-known that the group $\mathrm{Aut}(G)$ of automorphisms of the finite group $G$ acts on the data of the classification.
This is because our argument classifies nondegenerate braided fusion categories equipped with a fully faithful braided 2-functor from $\mathbf{2Rep}(G)$ whereas the statement of \Cref{FermionicNondegen} only considers the underlying nondegenerate braided fusion 2-categories.
There is however a more delicate redundancy at play here which arises from the fact that the braided fermionic strongly fusion 2-categories $\mathcal{S}$ and $\mathcal{T}$ have nontrivial automorphisms \cite{JF2,JFR}.
This subtle point already appeared in the classification of Morita equivalence classes of fermionic fusion 2-categories \cite{D9}, and was examined from a physical standpoint in \cite{Teixeira:2025qsg}.
Below, we explain how this can be seen from our methods.

For definitiveness, we treat the case when $\varsigma$ is trivial.
Then, \Cref{FermionicNondegen} proceeds by constructing a faithfully $G$-crossed braided extension of the braided fusion 2-category $\mathcal{S}=\mathcal{Z}(\mathbf{2SVect})$.
As explained in \Cref{rem:SandT} above, the braided fusion 2-category $\mathcal{S}$ has group of automorphisms $\mathrm{Aut}^{\br}(\mathcal{S})\cong\mathbb{Z}/16$.
Now, faithfully graded $G$-crossed braided extension of $\mathcal{S}$ are classified by homotopy classes of map $f:\mathrm{B}G\rightarrow \mathrm{B}\sPic(\mathcal{S})=\mathrm{B}^2\mathcal{S}^{\times}$.
But two such maps $f_1$ and $f_2$ will produce equivalent $G$-crossed braided fusion 2-categories if and only if there exists a braided automorphism $\varpi$ of $\mathcal{S}$ together with a homotopy witnessing the commutativity of the following diagram:
\begin{equation}
\begin{tikzcd}[sep=small]
\mathrm{B}G \arrow[r, "f_1"] \arrow[rd, "f_2"'] & \mathrm{B}^2\mathcal{S}^{\times}\,. \arrow[d, "\mathrm{B}^2\varphi", "\rotatebox{90}{$\simeq$}"'] \\
& \mathrm{B}^2\mathcal{S}^{\times}
\end{tikzcd}
\end{equation}
This observation explains the additional redundancies known to appear in the classification of fermionic nondegenerate braided fusion 2-categories.
With some work, they can be described explicitly in terms of supercohomology and we expect that the corresponding formulas coincide with those given in \cite[Theorem 2.9]{Teixeira:2025qsg}.
Finally, we wish to record that the action of both $\mathrm{Aut}^{\br}(\mathcal{S})$, or $\mathrm{Aut}^{\br}(\mathcal{T})$ depending on $\varsigma$, and $\mathrm{Aut}(G)$ completely capture the redundancies appearing in our classification result.
\end{rem}

\begin{ex}
We unpack the data of Theorem \ref{FermionicNondegen} for some explicit finite groups $G$. 
In order to demonstrate the existence of nontrivial (3+1)d fermionic TOs, a necessary condition is to show that the composition of $\varsigma \circ \tau$ is trivial.
We consider the case when $G\cong \Z/2^k$ and $\varsigma\in \SH^{5+\kappa}(\rB^2 \Z/2)$ is nontrivial.
Given a choice of $\tau\in \rH^2(\rB G;\Z/2)$, we indicate in the table below when the class $\varsigma \circ \tau\in\SH^{5+\tau}(\rB G)$ is trivializable.
Additionally, recall that $\varpi\in\SH^{4+\tau}(\rB G)$ parameterizes the choices of trivialization of $\varsigma \circ \tau$.
We collect these pieces of data in the following table:

$$\setlength{\tabcolsep}{4pt}
\renewcommand{\arraystretch}{1.2}
\begin{tabular}{c | c c c | c|c}
\hline\hline
$G_f$ & $G$ &  & $\tau$ & $\varsigma \circ \tau $ & $\varpi$ \\ 
\hline\hline 

$\Z/2^k\times \Z/2^F,\ k\geq 2$ 
& $\Z/2^k$ 
&  
& $0$ 
& \text{trivializable} 
& 0
\\ \hline 

$\Z/4^F$ 
& $\Z/2$ 
&  
& $x^2$ 
& \text{trivializable} 
&$\Z/2$
\\ \hline 

$\Z/(2^{k+1})^F,\ k\geq 2$ 
& $\Z/2^k$ 
&  
& $y$ 
& \text{non-trivializable} 
&$\Z/2$ 
\\ \hline  
\hline
\end{tabular}$$

\noindent 
In the table above, the class $x$  is the degree one generator for $\rH^*(\rB \Z/2; \Z/2)=\Z/2[x]$, and the class $y$ is the degree two generator for $\rH^*(\rB \Z/2^k; \Z/2) = \Z/2[x,y]/(x^2=0)$.

The class $\tau$ determines a central extension $G_f$ of $G$ by $\Z/2^F$, which we interpret as supplying fermion parity. The super-group $G_f$ corresponds to the global symmetry of the fermionic theory. When $\varsigma\circ\tau$ is trivializable, each choice of trivialization can be used as the Lagrangian for a Dijkgraaf-Witten construction of a fermionic (3+1)d TO.

We now briefly justify the data given above.
The nontrivial class $\varsigma \in \SH^{5+\kappa}(\rB^2 \Z/2)$ corresponds to the nontrivial element of $\rH^5(\rB^2 \Z/2;\bC^\times)$, as can be seen by considering the Atiyah-Hirzebruch spectral sequence to compute supercohomology. The group $\rH^5(\rB^2 \Z/2;\bC^\times)=\Z/2$ is generated by the class $(-1)^{T\Sq^1 T}$, where $T$ is the generator of $\rH^2(\rB^2 \Z/2;\Z/2)$. The composition $\varsigma \circ \tau$ is thus trivial when the pullback of $(-1)^{T\Sq^1 T}$ along the map $\rH^5(\rB^2 \Z/2; \bC^\times)\to \rH^5(\rB \Z/2^k; \bC^\times)$ is trivial. In the case when $T$ pulls back to $x^2$, this generator is trivialized.
On the other hand, when $T$ pulls back to $y$, this generator is not trivialized as can be seen by the Steenrod action.
\end{ex}

\subsection{Digression:\ The Classification of Braided Fusion 2-Categories}\label{subsection:allBraided}

We refine the strategy we used to classify nondegenerate braided fusion 2-categories so as to classify all braided fusion 2-categories.
In order to do so, we employ the methods of \S\ref{subsection:extensions} to construct $G$-crossed braided extensions of braided strongly fusion 2-categories. In contrast to the main results of \S\ref{subsection:BBF2C} and \S\ref{subsection:FBF2C}, the $G$-grading need not be faithful, thereby introducing additional difficulties.

Let $\mathfrak{A}$ be a braided fusion 2-category.
For simplicity, we will assume that $\mathfrak{A}$ is bosonic, so that $\Omega\mathfrak{A} = \mathbf{Rep}(G)$ as symmetric fusion 1-categories for some finite group $G$.
Now, consider $\mathcal{Z}_{(2)}(\mathfrak{A})$, the sylleptic center of $\mathfrak{A}$.
By definition, we have that $\Omega\mathcal{Z}_{(2)}(\mathfrak{A})$ is a symmetric fusion sub-1-category of $\Omega\mathfrak{A} = \mathbf{Rep}(G)$. We must therefore have $\mathcal{Z}_{(2)}(\mathfrak{A}) = \mathbf{Rep}(H)$, with $H$ a quotient of $G$. Let us write $N$ for the kernel of the projection $G\twoheadrightarrow H$.
The de-equivariantization of $\mathfrak{A}$ by $\mathbf{2Rep}(H)$ is a braided fusion 2-category $\mathfrak{A}_H$ equipped with an $H$-action by braided autoequivalences.
Moreover, we have that the de-equivariantization of $\mathfrak{A}_H$ by $\mathbf{2Rep}(N)$ is a faithfully $N$-graded $G$-crossed braided fusion 2-category $\mathfrak{C}$. Namely, we have that $\mathfrak{C}$ coincide with the de-equivariantization of $\mathfrak{A}$ by $\mathbf{2Rep}(G)$.
Furthermore, observe that $\mathfrak{C}$ is a bosonic strongly fusion 2-category.
In particular, its trivially graded component $\mathfrak{B}$ is a braided bosonic strongly fusion 2-category.

The above procedure can be reversed.
Begin with the data of a bosonic braided strongly fusion 2-category $\mathfrak{B}$.
Thanks to \Cref{thm:crossedextn}, we can classify all faithfully graded $N$-crossed braided extensions $\mathfrak{C}$ of $\mathfrak{B}$.
Then, the $N$-equivariantization of $\mathfrak{C}$ is a braided fusion 2-category $\mathfrak{C}^N$.
Equipping $\mathfrak{C}^N$ with an action of the finite group $H$ by braided autoequivalences, we can consider the braided fusion 2-category $\mathfrak{B}$ given by the $H$-equivariantization of $\mathfrak{C}^N$.
In order to ensure that $\Omega\mathfrak{B} = \mathbf{Rep}(G)$, the action of $H$ on $\mathfrak{C}^N$ must act on $\Omega\mathfrak{C}^N=\mathbf{Rep}(N)$ so that $\mathbf{Rep}(N)^H = \mathbf{Rep}(G)$.

Summarizing the above discussion, bosonic braided fusion 2-categories are classified by:
\begin{enumerate}
    \item A finite abelian group $E$;
    \item \label{number2} A class $\beta \in \rH^5(\rB^2 E; \mathbb{C}^\times)$;
    \item A finite group $N$;
    \item A map $f: \rB N \to \rB\sPic(\fB)$, i.e.\ a faithfully graded $N$-crossed braided extension $\mathfrak{C}$ of $\mathfrak{B} = \mathbf{2Vect}_E^{\beta}$;
    \item A finite group $H$;
    \item An action $\rho:\rB H \rightarrow \rB\sAut^{\br}(\fC^N)$ of $H$ on the braided fusion 2-category $\fC^N$, the $N$-equivariantization of $\mathfrak{C}$.
\end{enumerate}
The last point, regarding the action $\rho$, corresponds equivalently to the data of an action of $H$ on the $N$-crossed braided fusion 2-category $\mathfrak{C}$.
We also wish to note that the classification of fermionic braided fusion 2-categories is analogous. The only difference is that instead of a single class $\beta \in \rH^5(\rB^2 E; \mathbb{C}^\times)$, we require both a class $\kappa\in \rH^2(\rB^2 E;\mathbb{Z}/2)$ and a class $\varsigma \in \SH^{5+\kappa}(\rB^2 E)$.
 
We now explain how some familiar examples of braided bosonic fusion 2-categories are recovered using the above classification.

\begin{ex}
If $E$ is trivial and $N=G$, we recover the nondegenerate braided fusion 2-categories in \Cref{BosonicNondegen}. 
For a general finite abelian group $E$, we can choose a class $\beta\in \rH^5(\rB^2 E; \mathbb{C}^\times)$. Still working under the assumption that $N=G$, there is no choice of action $\rho$. We therefore only have the freedom of choosing a map $f:\rB G \rightarrow \rB\sPic(\fB)$.
If we assume that the map $f$ is trivial, we recover the braided fusion 2-categories of the form $\tVect^\beta_E\boxtimes \cZ(\tVect_G)$. Slightly more generally, if the map $f$ factors as $$\rB G\xrightarrow{\pi} \rB^4\mathbb{C}^{\times}\hookrightarrow \rB\sPic(\fB),$$ where $\rB^4\mathbb{C}^{\times}\hookrightarrow \rB\sPic(\fB)$ is the canonical inclusion, then the corresponding braided fusion 2-category is $\tVect^\beta_E\boxtimes \cZ(\tVect_G^{\pi})$.
\end{ex}

\begin{ex}
If $N$ is trivial, so that $G=H$, but $E$ is nontrivial, we only need to consider the action $\rho:\rB H \rightarrow \rB\sAut^{\br}(\fB)$. Suppose also that $\beta$ is trivial, then the homotopy groups of $\sAut^{\br}(\fB)$, where $\fB = \tVect_E$, are given by:
\begin{equation}\label{eq:groupsAutB}
    \begin{tabular}{|c|c|c|c|c|}
\hline
$\pi_0$ & $\pi_1$ & $\pi_2$ & $\pi_3$\\
\hline \\[-1em]
$\mathrm{Aut}(E)\ltimes \rH^4(\rB^2 E;\mathbb{C}^\times)$ & $0$ &  $\rH^2(\rB^2 E;\mathbb{C}^\times)$ & $0$ \\
\hline 
\end{tabular}\vspace{4mm}\,\,.
\end{equation}
The second entry is explained by $\pi_1(\sAut^{\br}(\fB))\cong\rH^3(\rB^2 E; \mathbb{C}^\times) = 0$. It is also convenient to rewrite $\pi_2(\sAut^{\br}(\fB))\cong\widehat{E}$, the Pontryagin dual of $E$. An action $\rho: \rB G \rightarrow  \mathrm{B}\mathscr{A}ut^{\br}(\mathfrak{B})$ therefore includes the data of a group homomorphism $G\rightarrow \mathrm{Aut}(E)\ltimes \rH^4(\rB^2 E;\mathbb{C}^\times)$. Assuming that this homomorphism factors through $\mathrm{Aut}(E)$, the only remaining data in $\rho$ is a class $\Theta \in H^3(\rB G;\underline{\widehat{E}})$, valued in coefficients that are twisted by the homomorphism $G\rightarrow \mathrm{Aut}(E)$. The class $\Theta$ can be viewed as a mixed anomaly between $H$ and $\widehat{E}$. Taking the equivariantization of the corresponding $G$-crossed braided fusion 2-category produces the braided fusion 2-category $\mathbf{2Rep}(\mathcal{G})$ of 2-representations of the 2-group $\cG$ classified by $G$, $\widehat{E}$, and the class $\Theta$.
\end{ex}

\begin{rem}
There are bosonic topological theories called \textit{mixed state} topological orders \cite{Wang:2023uoj,Ellison:2024svg,Yang:2025pke}. These occur in the presence of noise, e.g.\ interactions with background that disturb the pure state features. These topological order are not expected to occur in the ground
state of local gapped Hamiltonians. It was claimed in \cite{Sohal:2024qvq} that a partial classification of (2+1)d intrinsically mixed state topological orders is given by braided fusion 1-categories $\cC$ that are not nondegenerate. But, as $\cC$ corresponds to a bosonic topological theory, that is, a theory with no emergent fermions, its symmetric center must be Tannakian, and can therefore be de-equivariantized. This produces a nondegenerate braided fusion 1-category $\widetilde{\mathcal{C}}$, called the modularization of $\mathcal{C}$. In the other direction, the braided fusion category $\cC$ can be recovered as the equivariantization of its modularization $\widetilde{\mathcal{C}}$. Given the discussion in this section, which focuses on general braided fusion 2-categories, \Cref{thm:AllBoson} can be applied to classify mixed state topological orders in (3+1)d as will be discussed in \cite{KPY}. We see again that they all arise via equivariantization.
\end{rem}

\section{3+1D $G$-SETs and their anomalies}\label{section:G-TQFTs}

The topological orders arising from \Cref{def:toporder} are in general nonanomalous i.e.\ ``closed'' in the sense of \cite{Kong:2014qka,kong2015boundary,kong2017boundary}. They can however potentially harbor a gravitational anomaly. We explain here how to define a fermionic (3+1)d TO, in the sense of \Cref{def:degeneracy}, with an anomaly for a finite $G$-symmetry.
So as to motivate our approach, we find it instructive to first consider $G$-SETs in (2+1)d.

\subsection{Categorical Description of Fermionic (2+1)d Topological Order with $G$-Symmetry}\label{subsection:G2TQFT}

As reviewed in \S\ref{subsection:extensions}, equipping a (2+1)d TO with a $G$-symmetry amounts to performing a $G$-crossed braided extension of a nondegenerate braided fusion 1-category. For use below, let us additionally point out that it follows from \cite[Theorem 4.18(iii)]{DGNO}, that the condition of being fusion for a $G$-crossed braided multifusion 1-category corresponds to the condition that the braided functor from $\mathbf{Rep}(G)$ being fully faithful for the corresponding braided multifusion 1-category.

We now consider a fermionic version of this procedure. Given a fermionic (2+1)d TO described by a nondegenerate $\mathbf{2SVect}$-central braided fusion 1-category, we explain how to equip it with a $G$-symmetry. Mathematically, this process is captured by the notion of $\mathbf{2SVect}$-central $G$-crossed braided extension. The key difference with the bosonic case is that the $G$-extension is required to be compatible with the $\tSVect$-central structure.

In order to build towards the fermionic case, we turn our attention to braided rigid algebras in $\cZ(\tSVect_G)$.
This braided fusion 2-category is equivalent to $\cZ(\mathbf{2SRep}(G))$.
But, we have $\mathbf{2SRep}(G) \simeq \Mod (\mathbf{SRep}(G))$, so that braided rigid algebras in $\cZ(\mathbf{2SRep}(G))$ are precisely braided multifusion 1-categories equipped with a braided functor from $\mathbf{SRep}(G) = \mathbf{Rep}(G)\boxtimes \SVect$ thanks to \cite[Theorem 4.11]{davydov2021braided}.
The de-equivariantization of such a braided multifusion 1-category is a $G$-crossed braided multifusion 1-category $\mathcal{A} = \oplus_{g\in G}\mathcal{A}_g$ equipped with a $G$-\textit{equivariant} braided tensor functor $\SVect\rightarrow\mathcal{A}_e$. The key point here is that the $G$-equivariant structure on the functor ensures that, upon equivariantization, we obtain a braided monoidal functor from $\mathbf{SRep}(G)$. This is explained in detail in \cite[Theorem 3.13]{GVR}.

In order to describe fermionic (2+1)d $G$-SETs, we will make use of braided rigid algebras in the braided monoidal 2-category $\tSVect\boxtimes\cZ(\tVect_G)$,\footnote{This braided monoidal 2-category is equivalent to the centralizer of the braided embedding $\tSVect\rightarrow\cZ(\tSVect_G)$.} which are $\SVect$-central $G$-crossed braided multifusion 1-categories.
Unpacking the definition, we find that these are $G$-crossed braided multifusion 1-categories $\mathcal{A} = \oplus_{g\in G}\mathcal{A}_g$ equipped with a $G$-equivariant symmetric monoidal functor $\mathbf{SVect}\rightarrow \mathcal{Z}_{(2)}(\mathcal{A}_e)$.\footnote{Extension theory for central fusion 1-categories was studied in \cite{jones2022extension}. To our knowledge extension theory of central crossed braided fusion 1-categories has not yet been developed.}
Here we emphasize that the symmetric fusion 1-category $\mathbf{SVect}$ is implicitly endowed with the trivial $G$-action.
Then, fermionic (2+1)d $G$-SETs are precisely given by $\SVect$-central faithfully graded $G$-crossed braided fusion 1-categories that are nondegenerate in the sense that the functor $\mathbf{SVect}\rightarrow \mathcal{Z}_{(2)}(\mathcal{A}_e)$ is an equivalence.
Thanks to \cite[Proposition 4.56]{DGNO}, this nondegeneracy condition corresponds under (de-)equivariantization exactly to the requirement that the full symmetric fusion sub-1-category $\mathbf{SVect}\subset\mathbf{SRep}(G)$ of the equivariantization $\mathcal{A}^G$ coincides with the symmetric center.

\subsection{Categorical Description of Fermionic (3+1)d Topological Order with $G$-Symmetry}\label{subsection:GTQFT}

Following the 1-categorical line of reasoning presented above, the first step towards defining $\mathbf{2SVect}$-central $G$-crossed braided fusion 2-categories is the notion of a $G$-crossed braided fusion 2-category discussed in \S\ref{subsection:extensions}.
The second step is to consider braided rigid algebras in $\cZ(\mathbf{3SVect}_G)$.

\begin{lem}\label{lem:AlgInZ}
    Braided rigid algebras in $\cZ(\mathbf{3SVect}_G)$ are $G$-crossed braided multifusion 2-categories $\mathfrak{A} = \boxplus_{g\in G}\mathfrak{A}_g$ equipped with a $G$-equivariant braided monoidal 2-functor $\tSVect\rightarrow \mathfrak{A}_e$, where $\mathbf{2SVect}$ carries the trivial $G$-action.
\end{lem}
\begin{proof}
By definition, we have that $\mathbf{3SVect}_G\simeq \mathbf{3SVect}\boxtimes\mathbf{3Vect}_G$ as fusion 2-categories. But, $\mathbf{3Vect}_G$ is Morita equivalent to $\mathbf{3Rep}(G)$, so that $\mathbf{3SVect}_G$ is Morita equivalent to $\mathbf{3SRep}(G)\simeq \mathbf{3Rep}(G)\boxtimes \mathbf{3SVect}$, and therefore $\mathcal{Z}(\mathbf{3SVect}_G)\simeq \mathcal{Z}(\mathbf{3SRep}(G))$ as braided fusion 3-categories. In particular, the spaces of braided rigid algebras within these two braided fusion 3-categories are equivalent. This is a variant of (de-)equivariantization. Now, braided rigid algebras in $\mathcal{Z}(\mathbf{3SRep}(G))$ are precisely braided multifusion 2-categories $\mathfrak{B}$ equipped with a braided 2-functor $\mathbf{2SRep}(G)\rightarrow \mathfrak{B}$. The result therefore follows from the plain version of (de-)equivariantization given in Theorem \ref{thm:deequiv}, as the de-equivariantization of $\mathbf{2SRep}(G)$ is $\mathbf{2SVect}$ with the trivial $G$-action.
\end{proof}

After this initial observation, we move on to defining $\mathbf{2SVect}$-central $G$-crossed braided multifusion 2-categories.

\begin{defn}\label{def:enrichedGcrossed}
    A $\mathbf{2SVect}$-central $G$-crossed braided fusion 2-category is a strongly connected braided rigid algebra in $\mathbf{3SVect}\boxtimes \mathcal{Z}(\mathbf{3Vect}_G)$.
\end{defn}

We now present a $\mathbf{2SVect}$-central version of Theorem \ref{thm:deequiv}.

\begin{prop}\label{prop:application}
There is an equivalence of 3-categories between
$\mathbf{2SVect}$-central $G$-crossed braided fusion 2-categories and $\mathbf{2SVect}$-central braided fusion 2-categories equipped with a fully faithful braided 2-functor from $\mathbf{2Rep}(G)$.
\end{prop}
\begin{proof}
We begin by observing that $\mathbf{2SVect}$-central braided multifusion 2-categories equipped with a braided 2-functor from $\mathbf{2Rep}(G)$ are precisely rigid braided algebras in the braided fusion 3-category $\mathbf{3SVect}\boxtimes\mathcal{Z}(\mathbf{3Rep}(G))$.
But, there is an equivalence of braided fusion 3-categories
$$\mathbf{3SVect}\boxtimes \mathcal{Z}(\mathbf{3Vect}_G)\simeq \mathbf{3SVect}\boxtimes\mathcal{Z}(\mathbf{3Rep}(G))\,,$$
from which it follows that the 3-categories of braided rigid algebras are equivalent. Plainly, this establishes that there is an equivalence of 3-categories between $\mathbf{2SVect}$-central braided multifusion 2-categories equipped with a braided 2-functor from $\mathbf{2Rep}(G)$ and $\mathbf{2SVect}$-central $G$-crossed braided multifusion 2-categories.
The remaining part of the statement now follows exactly as in Theorem \ref{thm:deequiv} by unpacking what the condition of being strongly connected amounts to for a braided rigid algebra in either of these two braided fusion 3-categories.
\end{proof}

A variant of the proof of Proposition \ref{prop:nondegdeequiv} finally yields the desired result, which affords a description of fermionic (3+1)d TOs with $G$-symmetry.

\begin{prop}\label{prop:genuinelyfermionicTOclassification}
There is an equivalence of 3-categories between nondegenerate
$\mathbf{2SVect}$-central faithfully graded $G$-crossed braided fusion 2-categories and nondegenerate $\mathbf{2SVect}$-central braided fusion 2-categories equipped with a fully faithful braided 2-functor from $\mathbf{2Rep}(G)$.
\end{prop}

\begin{rem}
Nondegenerate $\mathbf{2SVect}$-central braided fusion 2-categories describe fermionic (3+1)d TOs, and, by \Cref{thm:fermionicTQFT}, are classified by a class in $\SH^4(\rB H)$.
For such braided fusion 2-categories, the data of a fully faithful braided 2-functor from $\mathbf{2Rep}(G)$ is precisely that of a surjective group homomorphism $H \twoheadrightarrow G$. We emphasize that the functor from $\tRep(G)$ does not factor through the sylleptic center. Namely, the sylleptic center only contains $\tSVect$ by nondegeneracy.
To finish recovering the data in \Cref{ansatz:fernionicTQFT}, we can think of the class in $\SH^4(\rB H)$ giving the action for the Dijkgraaf-Witten theory as equivalently labeling an invertible theory with $H$-symmetry. Then, writing $K \hookrightarrow H$ for the kernel of $H \twoheadrightarrow G$, there are $\SH^4(\rB K)$ choices of gauging for $K$.
\end{rem}

\begin{rem}
Let us now comment on the classification of (3+1)d $G$-SETs with an emergent fermion, but when the theory is bosonic.
It follows from \Cref{FermionicNondegen} that the theory without $G$-symmetry is described by a finite group $K$, a class $\varsigma\in \SH^{5+\kappa}(\rB^2\mathbb{Z}/2)$, $\tau \in \rH^2(\rB K;\Z/2)$, and $\varpi \in \SH^{4+\tau}(\rB K)$.
For simplicity, we will assume that $\varsigma$ is trivial.
In this case, we have argued in \Cref{rem:fermioniccenters} that the corresponding braided fermionic fusion 2-category is $\cZ(\tSVect^{(\tau,\varpi)}_K)$, where $\tSVect^{(\tau,\varpi)}_K$ is the fermionic strongly fusion 2-category classified by $\tau$ and $\varpi$.
In this case, there is an equivalence of group-like topological monoids $$\sPic\big(\cZ(\tSVect^{(\tau,\varpi)}_K)\big) \simeq \mathscr{B}r\mathscr{P}ic\big(\cZ(\mathbf{2Vect}^{(\tau,\varpi)}_{K})\big),$$
so that faithfully graded $G$-crossed braided extensions of $\cZ(\tSVect^{(\tau,\varpi)}_K)$ are in correspondence with faithfully $G$-graded extensions of $\tSVect^{(\tau,\varpi)}_K$.
Thanks to \cite{D11}, the latter extensions are classified by exact sequences of finite groups
\begin{equation}
    1 \longrightarrow K \longrightarrow H \longrightarrow G \longrightarrow 1,
\end{equation}
together with $\tau' \in \rH^2(\rB H;\Z/2)$ and $\varpi' \in \SH^{4+\tau'}(\rB H)$ such that $\tau'|_K = \tau$ and $\varpi'|_K = \varpi$.
\end{rem}

\subsection{Homotopical Description of Fermionic (3+1)d Topological Order with $G$-Symmetry}

We give a homotopically flavored description of $\mathbf{2SVect}$-central $G$-crossed braided fusion 2-categories by way of a $\mathbf{2SVect}$-central version of Theorem \ref{thm:crossedextn}. Physically, this corresponds to the process of equipping a fermionic (3+1)d TOs with $G$-defects.
This description will also be important for our subsequent discussion of fermionic $G$-anomalies in \S\ref{subsection:obstructiongauging}. In order to do so, we begin by recording the following unpacking of the definition of a $\mathbf{2SVect}$-central $G$-crossed braided fusion 2-category.

\begin{lem}\label{lem:enrichedcrossedbraided}
A $\mathbf{2SVect}$-central $G$-crossed braided fusion 2-category is a $G$-crossed braided fusion 2-category $\mathfrak{A} = \boxplus_{g\in G}\mathfrak{A}_g$ whose trivially graded component $\mathfrak{A}_e$ is equipped with a $G$-equivariant sylleptic 2-functor $\mathbf{2SVect}\rightarrow\mathcal{Z}_{(2)}(\mathfrak{A}_e)$, where $\mathbf{2SVect}$ carries the trivial $G$-action.
\end{lem}

Let us now fix a braided fusion 2-category $\mathfrak{B}$ equipped with a $G$-action, which we will express as a map of spaces $\rB G\rightarrow \rB \sAut^{\mathrm{br}}(\fB)$.
We give a homotopical description of the data of a $G$-equivariant structure on a sylleptic monoidal 2-functor $F:\mathbf{2SVect}\rightarrow\mathcal{Z}_{(2)}(\mathfrak{B})$.
By functoriality of the sylleptic center $\mathcal{Z}_{(2)}(-)$, there is a canonical map of spaces $\rB \sAut^{\mathrm{br}}(\fB)\rightarrow \rB\sAut^{\mathrm{syl}}( \cZ_{(2)}(\fB))$.
We will also consider the space $\sAut^{\mathrm{syl}}_{\tSVect}(F)$ of sylleptic monoidal autoequivalences of $F$, which is the space whose objects are pairs consisting of a sylleptic equivalence $E:\mathcal{Z}_{(2)}(\mathfrak{B})\simeq\mathcal{Z}_{(2)}(\mathfrak{B})$ together with a sylleptic natural equivalence $\phi$ witnessing the commutativity of the diagram
$$\begin{tikzcd}[sep=tiny]
&  & \mathcal{Z}_{(2)}(\mathfrak{B})\arrow[dd, "E"] \\
\mathbf{2SVect} \arrow[rru, "F"] \arrow[rrd, "F"'] &  & \\
&  & \mathcal{Z}_{(2)}(\mathfrak{B}).               
\end{tikzcd}$$
In particular, there is a map $\rB \sAut^{\mathrm{syl}}_{\tSVect}(F)\rightarrow \rB\sAut^{\mathrm{syl}}(\cZ_{(2)}(\fB))$ only recording the action on $\cZ_{(2)}(\fB)$. A $G$-equivariant structure on the sylleptic monoidal 2-functor $F$ then corresponds precisely to the data of a lift in the following diagram of spaces:
\begin{equation}
    \begin{tikzcd}[sep=small]
        && \rB\sAut^{\mathrm{syl}}_{\tSVect}(F) \arrow[dd]\\
        &&\\
        \rB G \arrow[r,"",swap ]  \arrow[uurr,dotted] & \rB \sAut^{\mathrm{br}}( \fB) \arrow[r] &\rB\sAut^{\mathrm{syl}}( \cZ_{(2)}(\fB)) \,.
    \end{tikzcd}
\end{equation}
For later use, it will be convenient to re-express the above definition more compactly. In order to do so, we define the space $\sAut^{\mathrm{br}}_{\mathbf{2SVect}}(\mathfrak{B})$ via the following pullback diagram:
$$\begin{tikzcd}[sep=tiny]
\arrow["\lrcorner"{anchor=center, pos=0.125}, draw=none, from=1-1, to=2-2]
\rB \sAut^{\mathrm{br}}_{\mathbf{2SVect}}(\mathfrak{B}) \arrow[rr]\arrow[dd]
& & \rB \mathscr{A}ut^{\mathrm{syl}}_{\tSVect
}(F) \arrow[dd]                     \\
& {} & \\
\rB \mathscr{A}ut^{\mathrm{br}}(\mathfrak{B}) \arrow[rr] & & \rB \mathscr{A}ut^{\mathrm{syl}}(\mathcal{Z}_{(2)}(\mathfrak{B})).
\end{tikzcd}$$
Then, a $G$-equivariant structure on $F$ corresponds to the data of a map $\rB G\rightarrow \rB\sAut^{\mathrm{br}}_{\mathbf{2SVect}}(\mathfrak{B})$.

Let now $\mathfrak{A}$ be a faithfully graded $G$-crossed braided extension of $\mathfrak{B}$, or, equivalently, the data of a map $\rB G\rightarrow \rB\sPic(\mathfrak{B})$. It follows from the above discussion that a $\mathbf{2SVect}$-central structure for the $G$-crossed braided extension $\mathfrak{A}$ corresponds precisely to a filler for the diagram below:
\begin{equation}
    \begin{tikzcd}[sep=small]
\rB G \arrow[dd] \arrow[rr] &                                            & \rB \mathscr{A}ut^{\mathrm{syl}}_{\tSVect
}(F) \arrow[dd]                     \\
& & \\
\rB \mathscr{P}ic(\mathfrak{B}) \arrow[r]                              & \rB \mathscr{A}ut^{\mathrm{br}}(\mathfrak{B}) \arrow[r] & \rB \mathscr{A}ut^{\mathrm{syl}}(\mathcal{Z}_{(2)}(\mathfrak{B})).
\end{tikzcd}
\end{equation}
This motivates the definition of the super Picard space $\mathscr{S}\!\mathscr{P}ic( \fB )$ of $\mathfrak{B}$ as the following pullback:

\begin{equation}\label{eq:enrichedPic}
\begin{tikzcd}[sep=small]
\arrow["\lrcorner"{anchor=center, pos=0.125}, draw=none, from=1-1, to=3-2]
\arrow["\lrcorner"{anchor=center, pos=0.125}, draw=none, from=1-2, to=3-3]
\rB \mathscr{S}\!\mathscr{P}ic(\mathfrak{B}) \arrow[dd] \arrow[r] &  \rB \sAut^{\mathrm{br}}_{\mathbf{2SVect}}(\mathfrak{B}) \arrow[r]\arrow[dd]
& \rB \mathscr{A}ut^{\mathrm{syl}}_{\tSVect
}(F) \arrow[dd]                     \\
& & \\
\rB \mathscr{P}ic(\mathfrak{B}) \arrow[r]                              & \rB \mathscr{A}ut^{\mathrm{br}}(\mathfrak{B}) \arrow[r] & \rB \mathscr{A}ut^{\mathrm{syl}}(\mathcal{Z}_{(2)}(\mathfrak{B})).
\end{tikzcd}
\end{equation}

\noindent From the above discussion, we deduce the following result.

\begin{thm}\label{thm:FGSET}
    Let $\fB$ be a $\mathbf{2SVect}$-central braided fusion 2-category. Then, $\mathbf{2SVect}$-central faithfully graded $G$-crossed braided extensions of $\fB$ are classified by homotopy classes of maps 
    \begin{equation}
        \rB G \rightarrow \rB \mathscr{S}\!\mathscr{P}ic( \fB )\,.
    \end{equation}
\end{thm}

\subsection{Obstruction to $G$-Gauging}\label{subsection:obstructiongauging}

We now turn our attention to quantifying the obstruction to gauging a $G$-symmetry on a  (3+1)d TO. In physics, gauging a global $G$-symmetry is a two step process consisting of coupling to background fields for the symmetry and summing over them in the path integral. A failure to couple to background fields results in a 't Hooft anomaly. Categorically, this anomaly translates into the obstruction to building a (faithfully graded) $G$-crossed braided extension.
Let us warm up by considering the case of a bosonic TO in (3+1)d equipped with a $G$-action. Mathematically, this corresponds to a nondegenerate braided fusion 2-category $\mathfrak{B}$ equipped with a $G$-action.
It is natural to ask whether there exists a faithfully graded $G$-crossed braided extension $\mathfrak{A} = \boxplus_{g\in G}\mathfrak{A}_g$ of $\mathfrak{B}$ for which the induced action on $\mathfrak{A}_e=\mathfrak{B}$ coincides with the given $G$-action on $\mathfrak{B}$. Homotopically, this is precisely the problem of lifting the map $\rho:\rB G \rightarrow \rB\mathscr{A}ut^{\mathrm{br}}(\mathfrak{B})$ representing the given $G$-action on $\mathfrak{B}$ to a map $\rB G \rightarrow \rB \mathscr{P}ic(\fB)$.
The obstruction to the existence of such lifts can be understood in a systematic way. In order to do so, we start by recalling a fiber sequence. Associated to any fusion $n$-category $\mathbf{C}$, there is a fiber sequence:
\begin{equation}\label{eq:fiberseq}
   \rB  \mathbf{C}^{\times}\longrightarrow  \rB\mathscr{A}ut^{\otimes}(\mathbf{C})\longrightarrow  \rB \mathbf{Bimod}(\mathbf{C})^{\times}\,.
\end{equation}
We outline a proof of this result in \Cref{sec:fiber}, noting that it was announced by Jones-Reutter but has not yet appeared.
Letting now $\mathbf{C}  = \Mod(\fB)$ for $\fB$ a braided fusion 2-category. This gives 
\begin{equation}\label{eq:SES2cat}
     \rB \sPic(\fB)\longrightarrow  \rB\mathscr{A}ut^{\mathrm{br}}(\mathfrak{B})\longrightarrow \rB \mathbf{Bimod}(\Mod(\fB))^{\times}\,.
\end{equation}
from which we see that the obstruction to lifting the map $\rho:\rB G \rightarrow \rB\sAut^{\br}(\fB)$ to a map  $ \rB G \rightarrow \rB \sPic(\fB)$ is given by $\rB \mathbf{Bimod}(\Mod(\fB))^{\times}$. This can be simplified using the fact that 
\begin{equation}\label{eq:Bimodrelation}
    \Bimod(\Mod(\fB)) \simeq \Mod(\cZ(\Mod(\fB))).
\end{equation}
As $\fB$ is nondegenerate, we have that $\mathbf{Mod}(\mathfrak{B})$ is invertible, so that $\mathcal{Z}(\mathbf{Mod}(\mathfrak{B}))\simeq\mathbf{3Vect}$, and 
$$\Bimod(\Mod(\fB))\simeq \mathbf{Mod}(\mathbf{3Vect})=\mathbf{4Vect}.$$
The corresponding space $\mathbf{4Vect}^{\times}$ coincides with the Witt space of nondegenerate braided fusion 1-categories introduced in \cite{DMNO}. We refer to the composite
$$\rB G \xrightarrow{\rho} \rB\sAut^{\br}(\fB)\rightarrow \rB \mathbf{4Vect}^{\times}$$
as the anomaly of $\rho$. Summarizing the above discussion, we have the following result.

\begin{prop}\label{prop:bosonicanomaly}
Let $\mathfrak{B}$ be a nondegenerate braided fusion 2-category equipped with a $G$-action $\rho$. Then, there exists a faithfully graded $G$-crossed braided extension of $\mathfrak{B}$ compatible with $\rho$ if and only if the anomaly of $\rho$ vanishes.
\end{prop}

We now turn our attention to quantifying the obstruction to gauging a $G$-symmetry on a fermionic TO in (3+1)d.
More precisely, given a fermionic TO corresponding to a nondegenerate $\mathbf{2SVect}$-central braided fusion 2-category $\mathfrak{B}$ equipped with a $G$-action, it is natural to ask whether there exists a $\mathbf{2SVect}$-central faithfully graded $G$-crossed braided extension $\mathfrak{A} = \boxplus_{g\in G}\mathfrak{A}_g$ of $\mathfrak{B}$ for which the induced action on $\mathfrak{A}_e=\mathfrak{B}$ coincides with the given $G$-action on $\mathfrak{B}$. Homotopically, this is precisely the problem of lifting the map $\rB G \rightarrow \rB\mathscr{A}ut^{\mathrm{br}}_{\mathbf{2SVect}}(\mathfrak{B})$ representing the given $G$-action on $\mathfrak{B}$ to a map $\rB G \rightarrow \rB \mathscr{S}\!\mathscr{P}ic(\fB)$ as in the following diagram 
\begin{equation}\label{eq:lifttoPic}
    \begin{tikzcd}
        & \rB \mathscr{S}\!\mathscr{P}ic( \fB )\arrow[d] \\
        \rB G \arrow[r] \arrow[ru,dotted]&  \rB \mathscr{A}ut^{\mathrm{br}}_{\mathbf{2SVect}}(\mathfrak{B})\,.
    \end{tikzcd}
\end{equation}

In order to do so, we will make use of the following fiber sequence, which is a $\mathbf{2SVect}$-central version of the fiber sequence presented above in Equation \eqref{eq:SES2cat}. In particular, this sequence involves the space $$\mathbf{4SVect}^{\times} := \mathbf{Mod}(\mathbf{Mod}(\mathbf{2SVect}))^{\times},$$ whose objects are Morita invertible $\mathbf{2SVect}$-central fusion 2-categories. Alternatively, thanks to \cite{D9}, one may also think of $\mathbf{4SVect}^{\times}$ as the space of Witt equivalence classes of nondegenerate $\mathbf{SVect}$-central braided fusion 1-categories as introduced in \cite{DNO}.

\begin{prop}\label{prop:superfibersequence}
Let $\mathfrak{B}$ be a nondegenerate $\mathbf{2SVect}$-central braided fusion 2-category. There is a fiber sequence
\begin{equation}\label{eqn:superfibersequence}
{\mathrm{B}\mathscr{S}\!\mathscr{P}ic(\mathfrak{B})}\rightarrow{\rB \mathscr{A}ut^{\mathrm{br}}_{\mathbf{2SVect}}(\mathfrak{B})}\rightarrow\rB \mathbf{4SVect}^{\times}.
\end{equation}
\end{prop}

The fiber sequence above follows a pattern similar to \cite[Section 5]{xu2025etale}, where super fusion 1-categories are considered.

As a preliminary result, we begin by identifying the space $$\mathbf{Bimod}(\Mod(\fB))^{\times}\simeq \mathbf{Mod}(\mathcal{Z}(\mathbf{Mod}(\fB)))^{\times}.$$

\begin{lem}\label{lem:SH6}
Let $\fB$ be a nondegenerate $\mathbf{2SVect}$-central braided fusion 2-category, there is an equivalence of braided fusion 3-categories $$\cZ(\Mod(\fB))\simeq \cZ(\mathbf{3SVect}).$$
\end{lem}
\begin{proof}
This follows from the fact, proven in \cite[Corollary V.4]{JF}, that every nondegenerate $\mathbf{2SVect}$-central braided fusion 2-category arises as the centralizer of $\tSVect$ in a nondegenerate braided fusion 2-category. It thence follows that $\fB$ has a minimal nondegenerate extension, so that there is a Morita equivalence between $\Mod(\fB)$ and $\mathbf{3SVect}$, and consequently an equivalence of braided fusion 3-categories $\cZ(\Mod(\fB))\simeq \cZ(\mathbf{3SVect})$ between their Drinfeld centers.
\end{proof}

The last result reduces our efforts to understanding the braided fusion 3-category $\mathcal{Z}(\mathbf{3SVect})$. We will need the following observation related to its structure.

\begin{lem}\label{lem:preenrichedfiebrsequence}
    There is a fiber sequence
    $$\rB^2\mathbf{3SVect}^{\times}\rightarrow \rB^2\mathcal{Z}(\mathbf{3SVect})^{\times}\rightarrow\rB\sAut^{\mathrm{br}}(\mathbf{3SVect})\simeq \rB^2\mathbb{Z}/2$$
    classifying the canonical action of $\sAut^{\mathrm{br}}(\mathbf{3SVect})$ on $\mathbf{3SVect}$.
\end{lem}
\begin{proof}
It follows from \Cref{lem:piPicbr} that $\pi_0(\mathcal{Z}(\mathbf{3SVect}))=\pi_0(\sPic^{\br}(\mathbf{2SVect})\cong\mathbb{Z}/2$.
More precisely, we have that $\mathcal{Z}(\mathbf{3SVect})$ contains an invertible object $G$ of order $2$.
The induced map of spaces $\rB^2\mathcal{Z}(\mathbf{3SVect})^{\times}\rightarrow\rB^2\mathbb{Z}/2$ must record the nontrivial action of $\sAut^{\mathrm{br}}(\mathbf{3SVect})\simeq \rB\mathbb{Z}/2$. Namely, otherwise the braided fusion 3-category $\mathcal{Z}(\mathbf{3SVect})$ would fail to be nondegenerate.
\end{proof}

We will also need the following refinement of the previous lemma.

\begin{lem}\label{lem:piBimod}
    There is a fiber sequence
    $$\rB\mathbf{Mod}(\mathbf{3SVect})^{\times}\rightarrow \rB\mathbf{Mod}(\mathcal{Z}(\mathbf{3SVect}))^{\times}\rightarrow\rB\sAut^{\mathrm{br}}(\mathbf{3SVect})\simeq \rB^2\mathbb{Z}/2.$$
\end{lem}
\begin{proof}
This will follow from Lemma \ref{lem:preenrichedfiebrsequence} once we have shown that the map $$\mathbf{Mod}(\mathbf{3SVect})^{\times}\rightarrow \mathbf{Bimod}(\mathbf{3SVect})^{\times}$$ induces a bijection on $\pi_0$.
In order to see this, note that objects of $\mathbf{Bimod}(\mathbf{3SVect})$ can be identified with Morita equivalence classes of multifusion 2-categories $\mathfrak{C}$ equipped with two (potentially distinct) braided monoidal 2-functor $\mathbf{2SVect}\rightarrow \mathcal{Z}(\mathfrak{C})$.
Thanks to \cite{D9}, these correspond precisely to Witt equivalence classes of braided multifusion 1-categories $\mathcal{B}$ equipped with two (potentially distinct) symmetric monoidal functor $\mathbf{SVect}\rightarrow\mathcal{Z}_{(2)}(\mathcal{B})$. Moreover, such braided multifusion 1-categories yield invertible bimodules, i.e.\ objects of $\mathbf{Bimod}(\mathbf{3SVect})^{\times}$, if and only if the two symmetric monoidal functors $\mathbf{SVect}\rightarrow\mathcal{Z}_{(2)}(\mathcal{B})$ are equivalences. But the group of autoequivalences of the symmetric fusion 1-category $\mathbf{SVect}$ is trivial. This concludes the proof as the objects of $\mathbf{Mod}(\mathbf{3SVect})^{\times}$ correspond to Witt equivalence classes of braided fusion 1-categories $\mathcal{A}$ equipped with an equivalence $\mathbf{SVect}\simeq\mathcal{Z}_{(2)}(\mathcal{A})$.
\end{proof}

\begin{proof}[Proof of Prop. \ref{prop:superfibersequence}]
Let $\mathfrak{B}$ be a nondegenerate $\mathbf{2SVect}$-central braided fusion 2-category. Given that we have assumed that $\mathfrak{B}$ is nondegenerate, the sylleptic 2-functor $F:\mathbf{2SVect}\rightarrow\mathcal{Z}_{(2)}(\mathfrak{B})$ is an equivalence. It follows that
$$\pt\simeq\rB\sAut^{\mathrm{syl}}_{\mathbf{2SVect}}(F)\rightarrow \rB\sAut^{\mathrm{syl}}(\mathbf{2SVect})\simeq \rB^2\mathbb{Z}/2.$$
Thanks to Lemma \ref{lem:piBimod}, there is a canonical map of spaces
$$\rB\mathbf{Bimod}(\mathbf{Mod}(\fB))^{\times}\simeq \rB\mathbf{Mod}(\mathcal{Z}(\mathbf{3SVect}))^{\times}\rightarrow\rB\sAut^{\mathrm{br}}(\mathbf{3SVect}).$$
We may therefore consider the following diagram of pullback squares, whose bottom row includes the fiber sequence in \Cref{eq:SES2cat}:
\begin{equation}\label{eq:pullbackfibersequences}
    \begin{tikzcd}[sep=small]
    {\mathrm{B}\mathscr{S}\!\mathscr{P}ic(\mathfrak{B})} & {\rB \mathscr{A}ut^{\mathrm{br}}_{\mathbf{2SVect}}(\mathfrak{B})} & {\rB \mathbf{4SVect}^{\times}} & {\mathrm{pt}} \\
	{\mathrm{B}\mathscr{P}ic(\mathfrak{B})} & {\mathrm{B}\mathscr{A}ut^{\mathrm{br}}(\mathfrak{B})} & {\mathrm{B}\mathbf{Bimod}(\mathbf{Mod}(\mathfrak{B}))^{\times}} & {\mathrm{B}^2\Z/2}. \\
	& {} & {}
	\arrow[from=1-1, to=1-2]
	\arrow[from=1-1, to=2-1]
	\arrow["\lrcorner"{anchor=center, pos=0.025}, draw=none, from=1-1, to=3-2]
	\arrow[from=1-2, to=1-3]
	\arrow[from=1-2, to=2-2]
	\arrow["\lrcorner"{anchor=center, pos=0.025}, draw=none, from=1-2, to=3-3]
	\arrow[from=1-3, to=1-4]
	\arrow[from=1-3, to=2-3]
	\arrow["\lrcorner"{anchor=center, pos=0.125}, draw=none, from=1-3, to=2-4]
	\arrow[from=1-4, to=2-4]
	\arrow[from=2-1, to=2-2]
	\arrow[from=2-2, to=2-3]
	\arrow[from=2-3, to=2-4]
\end{tikzcd}
\end{equation}
The identification of the right most pullback square is merely a reformulation of Lemma \ref{lem:piBimod}.
The descriptions of the remaining entries follow from the definitions together with the observation that the diagram of spaces
$$
\begin{tikzcd}[sep=small]
{\mathrm{B}\sAut^{\mathrm{br}}(\mathfrak{B})}\arrow[r]\arrow[d] & {\mathrm{B}\sAut^{\mathrm{syl}}(\mathcal{Z}_{(2)}(\mathfrak{B}))}\arrow[d, "\rotatebox{270}{$\simeq$}"]. \\
{\mathrm{B}\mathbf{Bimod}(\mathbf{Mod}(\mathfrak{B}))^{\times}}\arrow[r] & {\mathrm{B}\sAut^{\mathrm{br}}(\mathbf{3SVect})}.\\
\end{tikzcd}
$$
commutes as a consequence of the definitions.
\end{proof}

We see from Proposition \ref{prop:superfibersequence} that the space $\sWitt:=\mathbf{4SVect}^{\times}$ plays a special role. This space already appeared in \cite[Section 3.3]{JF2} where it was used to classify fermionic TOs up to gapped boundaries.\footnote{In \cite[Section 3.3]{JF2}, a connective super Witt spectrum was introduced. It was argued therein that this spectrum classifies fermionic TQFTs in all dimensions up to gapped boundary, Our definition of $\sWitt$ is a finite truncation of the super Witt spectrum of \cite{JF2}.}
The space $\sWitt$ has the following homotopy groups that can be expressed in terms of Morita classes of fermionic topological order:

\begin{itemize}
    \item $\pi_0\, \sWitt = s\mathcal{W}$: Morita equivalence classes of (2+1)d spin TQFTs
    \item $\pi_1\,\sWitt =0$
     \item $\pi_2 \, \sWitt=\Z/2$: In (1+1)d there is the trivial TQFT and Kitaev's Majorana chain
    \item $\pi_3 \, \sWitt=\Z/2$: In (0+1)d there is the trivial TQFT and the invisible fermion
     \item $\pi_4\, \sWitt =\mathbb{C}^\times$\,.
\end{itemize}

Let us denote by $\mathrm{SW}^5(\rB G)$ the abelian group of (homotopy classes of) maps from $\rB G$ to $\rB \sWitt$. This group can be used to give a mathematical definition of the 't Hooft anomaly of a fermionic TO in (3+1)d equipped with a $G$-symmetry.

\begin{defn}
Let $\mathfrak{B}$ be a nondegenerate $\mathbf{2SVect}$-central braided fusion 2-category $\mathfrak{B}$ equipped with a $G$-action represented by a map of spaces $\rho:\rB G\rightarrow\rB\sAut^{\mathrm{br}}_{\mathbf{2SVect}}(\mathfrak{B})$. The anomaly of the action $\rho$ is the class in $\mathrm{SW}^5(\rB G)$ given by the composite $$\rB G\xrightarrow{\rho} \rB\sAut^{\mathrm{br}}_{\mathbf{2SVect}}(\mathfrak{B})\rightarrow \rB \sWitt.$$
\end{defn}

By way of the fiber sequence in Equation \eqref{eqn:superfibersequence}, this yields the following result.

\begin{thm}\label{prop:sWittAnomaly}
Let $\mathfrak{B}$ be a nondegenerate $\mathbf{2SVect}$-central braided fusion 2-category equipped with a $G$-action $\rho$. Then, there exists a $\mathbf{2SVect}$-central faithfully graded $G$-crossed braided extension of $\mathfrak{B}$ compatible with $\rho$ if and only if the anomaly of $\rho$ vanishes.
\end{thm}

 In conclusion, we see that for fermionic (3+1)d TQFTs with $G$-symmetry, the classification of the obstruction to gauging the $G$-symmetry goes beyond what the symmetry extension construction in \cite{Wang:2017loc} provides. While there is a map from $\SH^5 \to \mathrm{SW}^5$, which implies there is a group cocycle part of the anomaly, there is also a layer for which a cocycle description is not clear. In order to complete the analogy with \cite{Wang:2017loc} in terms of the anomaly data, we note that it is possible that the homomorphism $\pi_1( \rB \sAut^{\br}_{\tSVect}(\fB)) \rightarrow \pi_1 (\rB \sWitt)$ is trivial, and hence the $G$-anomalies are classified by $\SH^5(\rB G)$. Physically, one can also take a limited approach in \Cref{ansatz:fernionicTQFT}  and only study a subgroup of the full obstruction, captured by $\SH^5(\rB G)$.  Then, the supercohomology data is exactly analogous to the data in \Cref{ansatz:WWW}. We refer the interested reader to \cite[Appendix A]{DYY1} for more details on the comparison. 
 In order to motivate why taking a limited approach is reasonable, we point out that the bosonic symmetry extension construction of \cite{Wang:2017loc}, which is claimed to work in any dimension, is itself already making a simplification by looking at cocycles in $\rH^{n+1}(\rB G; \mathbb{C}^\times)$. Even in the bosonic case, in high enough dimensions, there are anomaly contributions that go beyond ordinary cohomology. In (3+1)d, this is made apparent in \Cref{prop:bosonicanomaly} given that $\pi_0(\mathbf{4Vect}^{\times})\cong\Witt$ is the Witt group of nondegenerate braided fusion 1-categories. This last point is discussed in more detail in \cite[Section 3.3]{JF2}.

\subsection{Matching Anomalies with SETs}
With the theory of symmetry enriched topological orders at our disposal, we explain how it can be applied to the construction of anomalous (3+1)d fermionic TQFTs, extending the Wang-Wen-Witten construction. Consider the following table corresponding to the symmetry group $G = (\Z/2)^{\oplus 3}$ with anomaly in $\SH^5(\rB (\Z/2)^{\oplus 3})$. 

\begin{equation}
\renewcommand{\arraystretch}{1.4}
\begin{tabular}
{c c c | c|| c c c  c|| cc }
\hline \hline 
 $G$ &  &  & $\SH^{5}(\rB G)$ & $H$ &  &  & $\SH^4(\rB H)$ & $K$ & $\SH^4(\rB K)$ \\ \hline\hline 
 {$(\Z/2)^{\oplus 3}$} &  &  & $\text{at least}$ $\Z/2^{\oplus 2}$ & $(\Z/4)^{\oplus 3}$ & & & $(\Z/4)^{\oplus 8}\times (\Z/2)^{\oplus 4}$ & $(\Z/2)^{\oplus 3}$& $(\Z/2)^{\oplus 8}$ \\ \hline
\hline
\end{tabular}
\end{equation}

\noindent 
The groups $\SH^4(\rB H)$ and $\SH^4(\rB K)$ can be deduced from Theorem 11 and Theorem 3 respectively in \cite{Guo:2018vij}.
We consider an anomaly given by a class in a $(\Z/2)^{\oplus 2}$ subgroup of $\SH^{5}(\rB G)$.
The group $H=(\Z/4)^{\oplus 3}$ surjects onto $G$ and the  anomaly in the $(\Z/2)^{\oplus 2}$ subgroup of $\SH^{5}(\rB G)$ can be trivialized upon pulling back to $H$.
The choices of trivialization form a torsor over the group $\SH^{4}(\rB H)$.
The group $K$ on the right denotes the kernel of $H\twoheadrightarrow G$.
This group is associated to the TQFT constructed using the symmetry extension procedure. The map $\SH^{4}(\rB H)\rightarrow \SH^{4}(\rB K)$ is nonzero because the generators of $(\Z/4)^{\oplus 8}$ in $\SH^4(\rB H)$ map nontrivially to $(\Z/2)^{\oplus 8}$ in $\SH^{4}(\rB K)$. 

Following \Cref{ansatz:fernionicTQFT}, the above procedure constructs an anomalous fermionic theory. 
We can also generalize \Cref{ansatz:fernionicTQFT} to anomalies in twisted supercohomology $\SH^{5+\tau}(\rB G)$, and run an analogous construction.
As shown in \cite[Theorem 3.2]{DYY1}, any class $\pi \in \SH^{5+\tau}(\rB G)$ can be trivialized by pulling back along a finite group extension $H$ of $G$.
It is therefore always possible to realize a TQFT that saturates an anomaly given by a class in $\SH^{5+\tau}(\rB G)$. However, as we know from \cite{Freed:2016rqq}, anomalies for QFTs with emergent fermions i.e. those that couple to twisted spin structure  are classified by the cobordism group $\Hom(\Omega^{\mathrm{Spin}}_5(\rB G,\tau),\mathbb{C}^\times)$. Furthermore, it was observed in \cite[Table III]{Debray:2026sqw} that some of these anomalies are not captured by $\SH^{5+\tau}(\rB G)$.
The other main result of  \cite[Theorem 3.1]{DYY1} is that those anomalies in $\Hom(\Omega^{\mathrm{Spin}}_5(\rB G,\tau),\mathbb{C}^\times)$ that lie outside of supercohomology cannot be matched by a gapped topological theory. We direct the reader to \cite[Appendix A]{DYY1}  for an in depth discussion relating the three notions of anomaly provided by supercohomology $\mathrm{SH}$, $\mathrm{SW}$, and spin cobordism. We direct the reader to  \cite{Debray:2026pet} for learning how to compute supercohomology groups.

\subsection{Lagrangian Algebras in $\cZ(\mathbf{3Vect}_G)$ and $\cZ(\mathbf{3SVect}_G)$}\label{subsection:lagrangianalg}

Our classification results for (3+1)d $G$-SETs can be leveraged to explicitly describe Lagrangian algebras in certain nondegenerate braided fusion 3-categories associated to grouplike (4+1)d SymTFTs.
Here, we use the term nondegenerate in the context of \Cref{def:toporder} to mean that the braided fusion 3-category describing the (4+1)d SymTFTs has trivial sylleptic center.
When a topological phase of matter is described by a nondegenerate braided fusion category, inequivalent Lagrangian algebras correspond to distinct irreducible gapped boundaries for the phase.
In (2+1)d, this idea originates in \cite{BSS02,BSS03} and was subsequently fully realized in \cite{kong2014anyon}.
More precisely, given a (2+1)d theory of anyons described by a nondegenerate braided fusion 1-category $\cB$, a Lagrangian algebra in $\cB$ is a connected commutative separable algebra $A$ for which the braided fusion 1-category of \textit{local modules} results in the trivial phase, that is $\Mod^{\loc}_\cB(A)\simeq\Vect$.
Said differently, all the anyons that are not local with respect to the algebra $A$ get confined on the gapped boundary, which corresponds to $\Mod_\cB(A)$, and do not reside in the new phase described by $\Mod^{\loc}_\cB(A)$.
More recently, Lagrangian algebras for a (2+1)d SymTFT were used to classify symmetry breaking phases of a boundary theory \cite{Kong:2019byq,Kong:2019cuu,Kong:2021equ,Bhardwaj:2023fca}.
In (3+1)d, Lagrangian algebras in a nondegenerate braided fusion 2-categories are defined analogously \cite{ZLZHKT,DX}. More precisely, a connected braided separable algebra $A$ in a nondegenerate braided fusion 2-category is Lagrangian if the braided fusion 2-category $\Mod^{\loc}_\fB(A)$ of local modules is trivial.
The problem of classifying Lagrangian algebras in certain nondegenerate braided fusion 2-categories has also been undertaken in
\cite{xu:etale,Kong:2024ykr,Wen:2025thg}, which has had applications for classifying gapped and gapless phases of higher dimensional theories \cite{Bhardwaj:2024qiv,Bhardwaj:2023ayw}.
This motivates the following definition in (4+1)d.

\begin{defn}
    A connected rigid braided algebra $A$ in a nondegenerate braided fusion 3-category $\mathbf{B}$ is called Lagrangian provided that $\Mod^{\loc}_{\mathbf{B}}(A)\simeq \mathbf{3Vect}$.
\end{defn}

It was established in \cite[Corollary 3.3.4]{DX} that, for a braided fusion 1-category $\cB$, Lagrangian algebras in $\cZ(\Mod(\cB))$ correspond to nondegenerate braided fusion 1-categories $\cA$ equipped with a braided monoidal functor $\cB \rightarrow \cA$.
By mimicking their argument, we obtain an analogous classification of Lagrangian algebras in $\cZ(\mathbf{Mod}(\mathfrak{B}))$, where $\mathfrak{B}$ is a braided fusion 2-category.
Firstly, the categorification of \cite[Theorem 4.11]{davydov2021braided}, asserts that $\Mod^{\loc}(\fB) = \cZ(\Mod(\fB))$.
Moreover, for any braided fusion 3-category $\mathbf{B}$ and map of connected braided rigid algebras $B\rightarrow A$ in $\mathbf{B}$, we have $$\Mod^{\loc}_{\Mod^{\loc}_{\mathbf{B}}(B)}(A)\simeq \Mod^{\loc}_\mathbf{B}(A).$$ In particular, if we set $\mathbf{B} = \mathbf{3Vect}$, a connected braided rigid algebra $A$ in $\mathbf{B}$ is precisely a braided fusion 2-category $\fB$. In addition, such an algebra $A$ is Lagrangian if and only if $$\Mod^{\loc}(\fB) \simeq \cZ(\Mod(\fB)) \simeq \mathbf{3Vect}\,.$$ But, since $\cZ(\Mod(\fB))$ is nondegenerate, this is the case if and only if $$\Omega\cZ(\Mod(\fB))\simeq\mathcal{Z}_{(2)}(\fB) \simeq \tVect\,,$$ where the first step uses \cite[Section IV.B]{JF}. The above discussion yields the next result.

\begin{prop}\label{prop:LagrangianAlgebras}
    Let $\mathfrak{B}$ be a braided fusion 2-category, Lagrangian algebras in $\cZ(\mathbf{Mod}(\mathfrak{B}))$ correspond to nondegenerate braided fusion 2-categories equipped with a braided monoidal 2-functor from $\mathfrak{B}$.
\end{prop}

Taking $\fB = \tRep(G)$, the above result describes Lagrangian algebras in $\cZ(\mathbf{3Rep}(G)) = \mathcal{Z}(\mathbf{Mod}(\mathbf{2Rep}(G)))$. 
Such braided fusion 3-categories are particularly interesting because they serve as SymTFTs for grouplike symmetries in (3+1)d. Classifying their Lagrangian algebras is therefore a first step towards understanding the gapped phases for the boundary theory with such grouplike categorical symmetry, as prescribed in \cite{Bhardwaj:2023fca,Bhardwaj:2024qiv,Bhardwaj:2025piv} for (1+1)d and (2+1)d boundary theories. 
Using Theorems \ref{BosonicNondegen} and \ref{FermionicNondegen}, we can unpack the data of such Lagrangian algebras further.
Namely, as $\tRep(G)$ is a connected fusion 2-category, braided 2-functors $\tRep(G)\rightarrow\fA$ are completely determined by symmetric monoidal functors $\Rep(G)\rightarrow\Omega\fA$. If the nondegenerate braided fusion 2-category $\fA$ has all bosons, the corresponding Lagrangian algebras in $\cZ(\mathbf{3Rep}(G))$ are classified via the following data:
\begin{itemize}
    \item A finite group $H$ together with a group homomorphism $\varphi: H \rightarrow G$,
  \item A class $\pi\in\rH^4(\rB H; \mathbb{C}^\times)$.
\end{itemize}
If the nondegenerate braided fusion 2-category $\fB$ has emergent fermions, the corresponding Lagrangian algebras in $\cZ(\mathbf{3Rep}(G))$ are classified explicitly by:
\begin{itemize}
    \item A finite group $H$ together with a group homomorphism $\varphi: H\rightarrow G$,
    \item A class $\tau\in \rH^2(\rB H;\Z/2)$,
    \item A class $\varsigma \in \SH^{5+\kappa}(\rB^2\mathbb{Z}/2)$ such that $\varsigma\circ\tau$ is trivial, where $\kappa$ is the nontrivial class in $\rH^2(\rB^2\mathbb{Z}/2;\mathbb{Z}/2)$,
    \item A class $\varpi\in \SH^{4+\tau}(\rB H)$.
\end{itemize} 

\noindent Additionally, in the course of the proof of \Cref{thm:deequiv}, we have seen that there is an equivalence of braided fusion 3-categories $\cZ(\mathbf{3Vect}_G)\simeq \cZ(\mathbf{3Rep}(G))$.
In particular, the above classification of Lagrangian algebras in $\cZ(\mathbf{3Rep}(G))$ yields an analogous classification of Lagrangians algebras in $\cZ(\mathbf{3Vect}_G)$.

We now consider the case when the (4+1)d SymTFT is still grouplike, but has fermions.
This corresponds to considering the braided fusion 2-category $\fB = \mathbf{2Rep}(\widetilde{G},z)$. 
In order for a nondegenerate braided fusion 2-category to receive a map from $\mathbf{2Rep}(\widetilde{G},z)$, it must have emergent fermions, so that we only have to consider the nondegenerate braided fusion 2-categories classified in \Cref{FermionicNondegen}.
We then find that Lagrangian algebras in $\cZ(\mathbf{3Rep}(\widetilde{G},z))$ are classified explicitly by the following data:
\begin{itemize}
    \item A finite super-group $(\widetilde{H},z)$ together with a homomorphism $\varphi: (\widetilde{H},z)\rightarrow (\widetilde{G},z)$ of super-groups,
    \item A class $\varsigma \in \SH^{5+\kappa}(\rB^2\mathbb{Z}/2)$ such that $\varsigma\circ\tau$ is trivial, where $\kappa$ is the nontrivial class in $\rH^2(\rB^2\mathbb{Z}/2;\mathbb{Z}/2)$ and $\tau\in \rH^2(\rB H;\Z/2)$ classifies the central extension $\mathbb{Z}/2\rightarrow \widetilde{H}\rightarrow H$,
    \item A class $\varpi\in\SH^{4+ \tau}(\rB H)$.
\end{itemize}

\appendix

\section{A Fiber Sequence}\label{sec:fiber}

The purpose of this appendix is to justify the fiber sequence appearing in Equation \eqref{eq:fiberseq}.
Many variants thereof have already been considered \cite{GJS,davydov2021braided,JF2,Bhardwaj:2024xcx,Sa}.
In fact, a very general version of this fiber sequence was announced by Jones-Reutter but has not yet appeared.

We begin by recording the following technical result.

\begin{lem}\label{lem:technicalfiber}
Let $\mathcal{C}$ be an $(\infty,1)$-category.
Associated to any 1-morphism $f:\mathbbm{1}\rightarrow A$ in $\mathcal{C}$, there is a fiber sequence of spaces:
$$\rB \Omega_f\mathrm{End}_{\mathcal{C}}(\mathbbm{1},A)\rightarrow\rB \mathrm{End}_{\mathcal{C}_{\mathbbm{1}/}}(f)^{\times}\rightarrow\rB \mathrm{End}_{\mathcal{C}}(A)^{\times}\,.$$
\end{lem}
\begin{proof}
The forgetful map $p:\mathcal{C}_{\mathbbm{1}/}\rightarrow\mathcal{C}$ is a left fibration by \cite[Proposition 018F]{Ker}.
Viewing the 1-morphism $f:\mathbbm{1}\rightarrow A$ as an object in $\mathcal{C}_{\mathbbm{1}/}$, we can consider the endomorphism space $\mathrm{End}_{\mathcal{C}_{\mathbbm{1}/}}(f)$.
Further, we define a space $\mathrm{End}_{\mathcal{C}_{\mathbbm{1}/}}(f)_{\mathrm{id_A}}$ via the following pullback of spaces:
$$
\begin{tikzcd}[sep=small]
\mathrm{End}_{\mathcal{C}_{\mathbbm{1}/}}(f)_{\mathrm{id_A}} \arrow[r] \arrow[d] & \mathrm{End}_{\mathcal{C}_{\mathbbm{1}/}}(f) \arrow[d] \\
\{\mathrm{id}_A\} \arrow[r] & \mathrm{End}_{\mathcal{C}}(A)\,.
\end{tikzcd}
$$
As the right vertical arrow is a Kan fibration by \cite[Proposition 01P8]{Ker}, we obtain a fiber sequence of spaces
\begin{equation}\label{eq:intermediatefiber}
\mathrm{End}_{\mathcal{C}_{\mathbbm{1}/}}(f)_{\mathrm{id_A}}\rightarrow\mathrm{End}_{\mathcal{C}_{\mathbbm{1}/}}(f)\rightarrow\mathrm{End}_{\mathcal{C}}(A)\,.
\end{equation}

We now argue that $\mathrm{End}_{\mathcal{C}_{\mathbbm{1}/}}(f)_{\mathrm{id_A}}\simeq \Omega_f\mathrm{Hom}_{\mathcal{C}}(\mathbbm{1},A)$.
To see this, note that, by unpacking definition, we have 
$$\mathrm{End}_{\mathcal{C}_{\mathbbm{1}/}}(f)_{\mathrm{id_A}}\simeq \mathrm{Fun}(\Delta^2,\mathcal{C})\times_{\mathrm{Fun}(\partial\Delta^2,\mathcal{C})}\{F\}\,,$$
where $F:\partial\Delta^2\rightarrow\mathcal{C}$ is given by:
$$\begin{tikzcd}[sep=small]
& {A} \arrow[d, "\mathrm{id}_A"] \\
\mathbbm{1} \arrow[ru, "f"] \arrow[r, "f"'] & {A} 
\end{tikzcd}$$
But, using $\mathrm{Hom}^L_{\mathcal{C}}(\mathbbm{1},A)$ to denote the so-called left-pinched $\mathrm{Hom}$-space, we also have
$$\mathrm{Fun}(\Delta^2,\mathcal{C})\times_{\mathrm{Fun}(\partial\Delta^2,\mathcal{C})}\{F\}\simeq\Omega_f\mathrm{Hom}^L_{\mathcal{C}}(\mathbbm{1},A)\,.$$
The desired equivalence therefore follows from \cite[Proposition 01L5]{Ker}.

Finally, observe that the fiber sequence \eqref{eq:intermediatefiber} consists of topological monoids.
This follows for instance from the fact established in \cite[Proposition 2.2.5.1]{HTT} that the theory of $(\infty,1)$-categories can be modeled using simplicially enriched categories.
As a consequence, we obtain a fiber sequence on the associated group-like submonoids:
$$\mathrm{End}_{\mathcal{C}_{\mathbbm{1}/}}(f)_{\mathrm{id_A}}^{\times}\rightarrow\mathrm{End}_{\mathcal{C}_{\mathbbm{1}/}}(f)^{\times}\rightarrow\mathrm{End}_{\mathcal{C}}(A)^{\times}\,.$$
Delooping this last fiber sequence yields that appearing in the statement of the lemma.
\end{proof}

Let $\mathscr{V}$ be a monoidal $(\infty,1)$-category with geometric realizations and such that the monoidal product is compatible with geometric realizations. We write $\mathrm{Mor}_1(\mathscr{V})$ for the $(\infty,1)$-category of ($E_1$-)algebras, bimodules, and bimodule morphisms in $\mathscr{V}$ in the sense of \cite{Hau}.

\begin{prop}[The Jones-Reutter Fiber Sequence]\label{prop:JonesReutter}
Let $A$ be an algebra in $\mathscr{V}$. There is a fiber sequence of spaces:
$$\mathrm{B}\,\mathrm{Hom}_{\mathscr{V}}(\mathbbm{1},A)^{\times}
\rightarrow\mathrm{B}\mathrm{End}^{\otimes}_{\mathscr{V}}(A)^{\times}\rightarrow \mathrm{B}\mathbf{Bimod}_{\mathscr{V}}(A)^{\times}\,.$$
\end{prop}
\begin{proof}[Proof outline]
We apply \Cref{lem:technicalfiber} with $\mathcal{C}=\mathrm{Mor}_1(\mathscr{V})$, $A$ any algebra in $\mathscr{V}$, and $f$ the regular $\mathbbm{1}$-$A$-bimodule $A$.
With the above conventions, we have $$\Omega_f\mathrm{End}_{\mathcal{C}}(\mathbbm{1}) \cong \mathrm{End}_A(A_A)^{\times}\cong \mathrm{Hom}_{\mathscr{V}}(\mathbbm{1},A)^{\times}$$
as topological monoids.
This identifies the left most term of the fiber sequence.
Moreover, we have by definition 
$$\mathrm{End}_{\mathcal{C}}(A) = \mathbf{Bimod}_{\mathscr{V}}(A)^{\simeq}\,,$$
so that the right most term of the fiber sequence is as claimed.
It therefore suffices to identify the middle term.

In order to prove that $\mathrm{End}_{\mathcal{C}_{\mathbbm{1}/}}(f)^{\times}$ is equivalent to $\mathrm{End}^{\otimes}_{\mathscr{V}}(A)^{\times}$, it will suffice to show that
$$\mathrm{End}_{\mathcal{C}_{\mathbbm{1}/}}(A_A)\simeq\mathrm{End}^{\otimes}_{\mathscr{V}}(A)$$
as topological monoids.
In plain English, this is merely asserting that an $A$-$A$-bimodule $M$ in $\mathscr{V}$ that is equivalent to $A_A$ as a right $A$-module is precisely the data of an algebra homomorphism $\varphi:A\rightarrow A$ in $\mathscr{V}$.
Namely, a left $A$-module structure on $A_A$ compatible with the right action by $A$ is exactly the data of an algebra homomorphism $\varphi:A\rightarrow\mathrm{End}_A(A_A)\simeq A$. This argument can be traced back to the sketch of proof of \cite[Proposition 2.3.2]{JF2}.

We have to give a coherent version of argument of the previous paragraph.
In order to do so, we appeal to the following fact:
Let $M$ be an object of $\mathscr{V}$.
Assume that the $\mathscr{V}$-enriched endomorphism object $\underline{\mathrm{End}}(M)\in\mathscr{V}$ exists.
Then, it follows from \cite[Theorem 4.7.1.34]{HA} that $\underline{\mathrm{End}}(M)$ is an algebra in $\mathscr{V}$ and there is an equivalence between the space of left $A$-module structures on $M$ and the space of algebra homomorphisms $A\rightarrow\underline{\mathrm{End}}(M)$.
The above principle can be made relative.
Recall from \cite[Proposition 4.3.2.5]{HA} that $\mathscr{V}$ acts from the left on $\mathbf{Mod}_{\mathscr{V}}(A)$.
Let $N$ be a right $A$-module such that the $\mathscr{V}$-enriched endomorphism object $\underline{\mathrm{End}}_A(N)\in\mathscr{V}$ exists.
Again thanks to \cite[Theorem 4.7.1.34]{HA}, we have that $\underline{\mathrm{End}}_A(N)$ is an algebra in $\mathscr{V}$ and there is an equivalence between the space of left $A$-module structures on $N$ (as an object of $\mathbf{Mod}_{\mathscr{V}}(A)$) and the space of algebra homomorphisms $A\rightarrow\underline{\mathrm{End}}_A(N)$.

We now consider $N=A_A$, the regular right $A$-module.
Thanks to \cite[Corollary 4.2.4.6]{HA}, we have $\underline{\mathrm{End}}_A(A)\simeq A$ canonically as algebras in $\mathscr{V}$.
In particular, the discussion of the previous paragraph applies, so that there is an equivalence of spaces:
$$\mathbf{LMod}_{\mathbf{Mod}_{\mathscr{V}}(A)}(A)^{\simeq}\times_{\mathbf{Mod}_{\mathscr{V}}(A)^{\simeq}}\{A_A\}\simeq\mathrm{End}_{\mathscr{V}}^{\otimes}(A)\,.$$
But, there are also equivalences of spaces
$$\mathrm{End}_{\mathcal{C}_{\mathbbm{1}/}}(A_A)\simeq \mathbf{Bimod}_{\mathscr{V}}(A)^{\simeq}\times_{\mathbf{Mod}_{\mathscr{V}}(A)^{\simeq}}\{A_A\}\simeq \mathbf{LMod}_{\mathbf{Mod}_{\mathscr{V}}(A)}(A)^{\simeq}\times_{\mathbf{Mod}_{\mathscr{V}}(A)^{\simeq}}\{A_A\}\,.$$
The first one holds by unpacking the definitions, whereas the second is \cite[Theorem 4.3.2.7]{HA}.
Combining these three equivalences together, we find that there is an equivalence $\mathrm{End}_{\mathcal{C}_{\mathbbm{1}/}}(A_A)\simeq\mathrm{End}^{\otimes}_{\mathscr{V}}(A)$ as spaces.
One checks that the above equivalence is in fact inverse to the functor $\mathrm{End}_{\mathscr{V}}^{\otimes}(A)\rightarrow\mathrm{End}_{\mathcal{C}_{\mathbbm{1}/}}(A_A)$ sending a map of algebras $\varphi:A\rightarrow A$ to the $A$-$A$-bimodule $_{\varphi}A_A$.
This last assignment is compatible with the monoid structures, which concludes the outline of the proof.
\end{proof}

The proof outline we give above for Proposition \ref{prop:JonesReutter} holds in more generality:\
We only use that $\mathscr{V}$ has geometric realization and that its monoidal structure is compatible with them so as to guarantee the existence of the relative tensor products of bimodules over arbitrary algebras in $\mathscr{V}$.
Alternatively, we can restrict our attention to a specific class of algebras in $\mathscr{V}$ for which such relative tensor products exist.
We are specifically interested in the symmetric monoidal $n$-category $\mathbf{nVect}$ of complex $n$-vector spaces with $n\leq 3$.
As posited in \cite{Gaiotto:2019xmp} see also \cite[Section II]{JF}, the (symmetric monoidal) $(n+1)$-category of $(n+1)$-vector spaces coincides with the Morita category on the condensation algebras in $\mathbf{nVect}$.
The case $n=1$ was established in \cite{D1}, and the case $n=2$ in \cite{D8}. Some discussion of the case $n=3$ appears in \cite{Bhardwaj:2024xcx} to which we refer the reader for additional details.
The above discussion yields the following result which had already been recorded in \cite[Theorem 5.2.24]{Bhardwaj:2024xcx} and \cite[Proposition 2.3.2]{JF2}:

\begin{cor}\label{cor:thefibersequence}
Let $\mathbf{C}$ be a fusion $n$-category, there is a fiber sequence:
$$\mathrm{B}\mathbf{C}^{\times}\rightarrow\mathrm{B}\mathscr{A}ut^{\otimes}(\mathbf{C})\rightarrow\mathrm{B}\mathbf{Bimod}(\mathbf{C})^{\times}\,.$$
\end{cor}

\section*{Acknowledgments}

It is our pleasure to thank
Andrea Antinucci,
Arun Debray,
Ryohei Kobayashi,
Dmitri Nikshych,
Abhinav Prem,
Sakura Schäfer-Nameki,
Juven Wang, and
Weicheng Ye
for helpful discussions regarding the content of this paper.
In addition, we would like to thank the referee for careful reading our manuscript and suggesting many improvements.
TD is supported by the Simons Collaboration on Global Categorical Symmetries. He would also like to thank the University of Oxford where part of this work was completed. 
MY is supported by the EPSRC Open Fellowship EP/X01276X/1. He would also like to thank Harvard University, where part of this work was completed.
This research was supported in part by grant NSF PHY-2309135 to the Kavli Institute for Theoretical Physics (KITP). Part of this work
was completed at the KITP and benefited from interactions with the participants of the  program “Generalized
Symmetries in Quantum Field Theory:\ High Energy Physics, Condensed Matter, and Quantum
Gravity”. 

\section*{Conflict of Interest and Data Availability}
All authors certify that they have no affiliations with or involvement in any organization or entity with any financial interest or non-financial interest in the subject matter or materials discussed in this manuscript.  There is no data available
for this article to declare.
\bibliographystyle{alpha}
\bibliography{ref.bib}
\end{document}